\documentclass[a4paper,fleqn]{format}
\usepackage[authoryear,longnamesfirst]{natbib}
\usepackage{xcolor,romannum,hyperref,cleveref,enumitem}
\usepackage{amsmath}
\usepackage{amsthm}
\usepackage{mathtools}
\usepackage{placeins}
\usepackage{longtable}
\usepackage{booktabs}
\DeclarePairedDelimiter{\norm}{\lVert}{\rVert}
\newcommand{\red}[1]{\textcolor{black}{#1}}

\newtheorem{proposition}{Proposition}
\usepackage{makecell}
\usepackage{capt-of}
\crefname{enumi}{}{}
\Crefname{enumi}{}{}
\crefname{enumii}{}{}
\Crefname{enumii}{}{}
\crefname{enumiii}{}{}
\Crefname{enumiii}{}{}
\crefname{enumiv}{}{}
\Crefname{enumiv}{}{}
\crefname{equation}{Eq.}{Eqs.}
\usepackage{chngcntr}

\def\tsc#1{\csdef{#1}{\textsc{\lowercase{#1}}\xspace}}
\tsc{WGM}
\tsc{QE}
\tsc{EP}
\tsc{PMS}
\tsc{BEC}
\tsc{DE}
\begin{document}
\let\WriteBookmarks\relax
\def\floatpagepagefraction{1}
\def\textpagefraction{.001}
\shorttitle{Hierarchical Graph Learning for Calendar Spread Strategies in Commodity Futures Markets}
\shortauthors{Y. Hong and D. Klabjan}
%\begin{frontmatter}
\renewcommand{\thepage}{\arabic{page}}

\title [mode = title]{Hierarchical Graph Learning for Calendar Spread Strategies in Commodity Futures Markets}

\author[1]{Yoonsik Hong}[% type=editor,
                        % auid=000,bioid=1,
                        % prefix=Sir,
                        % role=Researcher,
                        orcid=0009-0008-2002-7266]
\cormark[1]
\ead{YoonsikHong2028@u.northwestern.edu}
\credit{Conceptualization, Data curation, Investigation, Methodology, Writing---original draft, Writing---review and editing, and Software}

\affiliation[1]{organization={Department of Industrial Engineering and Management Sciences, Northwestern University},
                city={Evanston},
                postcode={60208}, 
                state={IL},
                country={USA}}[foot]

\author[1]{Diego Klabjan}[]
\ead{d-klabjan@northwestern.edu}
\credit{Supervision, Methodology, Writing---review and editing, and Resources}

\cortext[cor1]{Corresponding author}

\begin{abstract} 
Commodity futures can be represented hierarchically, with underlying assets at the upper level and individual futures contracts at the lower level.
Entities at each level can be connected by edges reflecting inherent correlations, with cross-level edges capturing contract-to-underlying asset connections.
Building on our observations of these structures, we propose a hierarchical graph learning approach for calendar spread (CS) strategies in commodity futures markets, addressing two significant gaps in the machine-learning literature: (i) the absence of learning-based methods for CS strategies in futures markets, and (ii) the lack of consideration of maturity-dependent interrelationships across commodity futures. 
We first establish the efficacy of CS strategies by analytically showing that CS strategies can possess higher risk-adjusted returns, measured by the information ratio, and lower risk, measured by variance and delta, than long-only strategies. 
We then introduce a method to convert learning-based predictions into CS positions. 
Next, we develop a hierarchical graph learning method that predicts futures price movements by utilizing the maturity-dependent interrelationships, thereby yielding a CS trading algorithm. 
Empirical results on commodity futures markets traded on the Chicago Mercantile Exchange Group demonstrate that our method outperforms benchmark models in both prediction and trading performance. 
We find that maturity-dependent interrelationships across commodity futures are instrumental in prediction and that CS trading based on hierarchical graph learning is effective for statistical arbitrage. 
\end{abstract}


\maketitle

\setcounter{page}{1} %I added

\section{Introduction}

We address the problem of exploiting statistical arbitrage opportunities in commodity futures markets through a calendar spread (CS) strategy based on hierarchical graph learning that captures information embedded in maturity-dependent interrelationships across commodity futures. 
A futures contract is an agreement between two parties in which the contract price is fixed at initiation, and at maturity, one party delivers the underlying asset (or settles in cash), while the other pays the predetermined price \citep{hull2013fundamentals}.
If the underlying asset is a commodity, the contract is referred to as a commodity futures contract. 
Statistical arbitrage refers to cash flows that generate profits by exploiting mispricing under acceptable risk \citep{lazzarino2018statistical}, and deep-learning–based approaches for capturing such opportunities have been actively studied across various financial markets (e.g., \cite{guijarro2025deep,hong2025graph}). 
To exploit statistical arbitrage in commodity futures markets, we employ a CS strategy that takes long and short positions in futures contracts with the same underlying commodity but different times to maturity (TTMs).

No prior learning-based studies address CS trading, despite its widespread use and extensive empirical studies in futures markets \citep{clarke2013fundamentals,szymanowska2014anatomy,boons2019basis,CMEFuturesSpreads}. Trading strategies in existing learning-based futures studies---such as single-contract directional trading and top-K selection---are highly likely to be exposed to underlying asset price risk. 
Such exposure to volatile underlying price movements, rather than relatively stable cost-of-carry components (e.g., interest rates and storage costs), can deteriorate risk-adjusted trading performance. 
Motivated by this limitation, we show that CS strategies can mitigate risk and yield superior risk-adjusted returns.
We also introduce a method to convert learning-based predictions into CS positions.

Most existing learning-based studies overlook TTM-dependent interrelationships across commodity futures. 
Commodity futures contracts are economically linked \citep{chng2009economic,vacha2013time,cai2020co,furlong1996commodity,hammoudeh2009relationships,boakye2024commodity,liu2023nonlinear}, but this structure is rarely incorporated into learning-based models. 
Although some recent studies \citep{hu2025graph,tan2024futures} adopt graph learning to model relationships among commodity futures, they do not distinguish these relationships by TTM.
However, these relationships are inherently TTM-dependent \citep{brennan1976supply,gibson1990stochastic,schwartz1997stochastic,casassus2005stochastic}. 
For example, gasoline and diesel inventories may remain low for three months due to a war, but they are expected to recover after six months as supply chains recover. 
When predicting the price of diesel futures with a six-month TTM, the price information of gasoline futures with a three-month TTM should not be treated in the same way as that of gasoline futures with a six-month TTM. 
Ignoring TTM thus limits the ability of graph-learning models to capture TTM-dependent information. 
To address this gap, we propose a hierarchical graph learning approach.

First, we verify the efficacy of CS strategies in terms of risk and risk-adjusted return.
We theoretically establish their advantages over long-only (LO) strategies. 
Since any trading position can be expressed as a linear combination of CS and LO positions, we compare the two strategies. 
We show that CS strategies possess lower variance and risk exposure than LO strategies under some assumptions. 
Moreover, we show that if an LO strategy outperforms an uninformed strategy by a certain threshold, then a CS strategy derived from the LO strategy outperforms the LO strategy in terms of risk-adjusted return. 
We also present a method to assess whether the assumptions are realistic. 
We further introduce a method to map learning-based predictions into CS positions.

Next, we develop a hierarchical graph learning approach that captures maturity-dependent relationships across commodity futures. 
We construct a bi-level graph in which underlying commodities form upper-level nodes and individual futures contracts form lower-level nodes, connected via cross-commodity, commodity-contract, and cross-contract relationships. 
Since commodities do not share a common contract maturity grid, 
we introduce a graph lifting method that constructs TTM-aligned virtual futures nodes for each commodity on a virtual TTM grid shared across commodities.
On this lifted graph, we design a bi-level convolution architecture that alternates between (i) propagating information across commodities through virtual futures at the same TTM (TTM-aligned propagation) and (ii) capturing intra-commodity term structure dynamics via message passing over neighboring maturities (different-TTM propagation). 
By distinguishing relationships based on TTM through distinct parameterizations and iteratively integrating information across both dimensions, the proposed method captures interrelationships across commodity futures while explicitly accounting for maturity dependence.

Our experiments demonstrate that the proposed method generates statistical arbitrage and that both CS strategies and maturity-dependent interrelationships are effective. 
We provide evidence that CS strategies can exhibit lower variance, lower exposure risk, and higher IRs than LO strategies in commodity futures markets.
Our hierarchical graph learning achieves the lowest mean squared error with statistical significance, and its CS trading attains the best risk-adjusted returns by capturing TTM-dependent interrelationships among commodity futures, compared to other learning-based methods that either ignore such interrelationships or do not account for their maturity dependence.

This paper makes the following main contributions.
\begin{itemize}
\item To the best of our knowledge, this study is the first to propose a hierarchical graph learning approach for predicting commodity futures price movements while accounting for maturity-dependent interrelationships across commodity futures. 
\item We establish the efficacy of calendar spread strategies by presenting propositions showing that they can exhibit higher information ratios, lower variance, and lower exposure risk than long-only strategies in commodity futures markets.
\item Empirical results show that the hierarchical graph learning method outperforms benchmark methods with statistical significance in prediction tasks. 
We find that maturity-dependent interrelationships across commodity futures are instrumental in predictive performance. 
\item Our calendar spread strategy attains daily information and Sortino ratios of 0.0846 and 0.1241, corresponding to improvements of 96\% and 82\% over the benchmarks, and 75\% and 113\% over the S\&P 500, respectively. 
These results indicate that calendar spread trading based on hierarchical graph learning in futures markets is effective for statistical arbitrage. 
\end{itemize}

\section{Related Work}
No learning-based studies address CS strategies in futures markets. 
Existing work applies various learning-based methods to prediction and trading, including LSTM and GRU models \citep{li2022novel,zhang2020energy,hung2021dpp}, econometric regression \citep{wang2024predictability,angelidis2025predicting}, reinforcement learning \citep{gabrielsson2015high,massahi2024deep,kaur2025hybrid}, clustering-based approaches \citep{fengqian2020adaptive}, and AI-agent systems incorporating vision-language models \citep{wang2025agricultural}. 
These studies span agricultural futures \citep{li2022novel,liu2023nonlinear,wang2025agricultural,fengqian2020adaptive}, energy futures \citep{zhang2020energy,sun2019predicting,hu2025graph,barunik2016forecasting}, metal futures \citep{massahi2024deep,kaur2025hybrid,fengqian2020adaptive,sun2019predicting,tan2024futures}, equity index futures \citep{hung2021dpp,gabrielsson2015high,fengqian2020adaptive}, and broad cross-market commodity universes \citep{wang2024predictability,angelidis2025predicting}. 
Many learning-based studies focus solely on price or return prediction \citep{li2022novel,zhang2020energy,hung2021dpp}, while trading-oriented work considers single-contract directional trading \citep{gabrielsson2015high,fengqian2020adaptive,wang2025agricultural,massahi2024deep,kaur2025hybrid,sun2019predicting}, top-K trading \citep{hu2025graph}, cross-commodity long-short trading \citep{wang2024predictability,angelidis2025predicting}, or cross-commodity pairs trading \citep{liu2023nonlinear}. 
However, none of the prior learning-based work addresses CS strategies. 
We conjecture that a key reason is the lack of clear theoretical guarantees on the efficacy of CS strategies, despite their widespread empirical use in futures markets  \citep{clarke2013fundamentals,szymanowska2014anatomy,boons2019basis,CMEFuturesSpreads}.
Accordingly, we provide theoretical justification and demonstrate that CS trading based on hierarchical graph learning is effective compared with benchmark methods.

Existing learning-based studies largely overlook TTM-dependent interrelationships across commodity futures. 
Futures contracts across different commodities are economically interconnected through production chains \citep{chng2009economic}, substitution and complementarity effects \citep{vacha2013time}, and shared macroeconomic drivers \citep{cai2020co,furlong1996commodity,hammoudeh2009relationships,boakye2024commodity,liu2023nonlinear} (e.g., oil and agricultural markets \citep{ji2012does,mitchell2008note,nazlioglu2011world}), yet most learning-based approaches treat contracts in isolation.  
Although a few recent studies \citep{hu2025graph,tan2024futures} incorporate cross-futures relationships through graph learning, they treat commodity futures contracts as generic nodes and construct connections without conditioning on TTM.  
However, interrelationships across commodity futures are inherently TTM-dependent, as contracts with different TTMs exhibit heterogeneous sensitivities to inventory conditions, short-term supply disruptions, production constraints, and long-run demand and substitution dynamics \citep{brennan1976supply,gibson1990stochastic,schwartz1997stochastic,casassus2005stochastic}. 
Ignoring TTM thus conflates distinct economic channels and limits the ability of graph architectures to capture TTM-dependent relationships across commodity futures. 
Their methods \citep{hu2025graph,tan2024futures} also cannot address node birth and death, as their focus on high-frequency trading assumes a fixed contract universe. 
In contrast, we consider daily trading, where node birth and death arise naturally from contract maturities. 
Through our hierarchical graph learning method, we highlight the importance of TTM dependence in commodity futures interrelationships and demonstrate the instrumental role of such interrelationships.

\section{Notation and Preliminaries}
\label{subsec:notation}

For a finite set $\mathcal{I}$ and vectors $\mathbf{x}_i \in \mathbb{R}^p$ for $i \in \mathcal{I}$, let $[\mathbf{x}_i]_{i \in \mathcal{I}} \in \mathbb{R}^{|\mathcal{I}| \times p}$ denote the row-wise stacked matrix in lexicographical order of $\mathcal{I}$.
For an index set $\mathcal{I}$, if $x_i \in \mathbb{R}$ is defined for
$i \in \mathcal{I}' \subset \mathcal{I}$ and not defined for
$i \in \mathcal{I} \setminus \mathcal{I}'$, we let
$\mathrm{sum}(x_i; i \in \mathcal{I})$, $\mathrm{avg}(x_i; i \in \mathcal{I})$ and
$\mathrm{std}(x_i; i \in \mathcal{I})$ denote the sum, average, and standard deviation over $\mathcal{I}'$, respectively. 
We denote by $\mathbb{I}_{\{\cdot\}}$ the indicator function. 
For $a \leq b$, let $[a:b]$, $[a:b)$, $(a:b]$, and $(a:b)$ denote the corresponding integer intervals. 
For $\mathbf{x} =[x_i]_{i\in[1:n]}\in\mathbb{R}^n$, we define $\operatorname{cdf}(x;\mathbf{x})=\frac{1}{n}\sum_{i=1}^n\mathbb{I}_{\{x_i \leq x\}}$ and $\widehat{\operatorname{cdf}}(x;\mathbf{x})=\frac{1}{n+\epsilon}\sum_{i=1}^n\mathbb{I}_{\{x_i \leq x\}}$  where $\epsilon > 0$ is a small constant.  
For any random variable $X$, let $\tilde X$ denote a realization and $\hat X$ denote its mean estimate. 
Let $\mathbf{1}$ denote a vector of ones.
For $x \in \mathbb{R}$, define $[x]^+ = \max\{0,x\}$.
We let $W_\cdot$ and $\mathbf b_\cdot$ denote a real-valued learnable parameter matrix and a real-valued learnable bias vector, respectively, where the subscripts indicate their indices, and let $\phi$ denote an activation function.

We index trading dates by $t \in \mathbb{N}$ in chronological order, without skipped integers. 
Let $\mathcal{C} \subset \mathbb{N}$ denote the index set of commodities, with $\max \mathcal{C} = |\mathcal{C}|$.
For each $c \in \mathcal{C}$, let $\mathcal{D}_{c}\subset \mathbb{N}$ denote the index set of futures contracts with underlying commodity $c$, ordered by increasing maturity, and satisfying $\max \mathcal{D}_c = |\mathcal{D}_c|$.
We denote a futures contract by $(c,d) \in \mathcal{C} \times \mathcal{D}_c$  with maturity $T_{cd}$. 
We denote by $\mathcal{T}_c=\{T_{cd}:d\in\mathcal{D}_{c}\}$ the maturity grid of commodity $c\in\mathcal{C}$. 
We refer to $T_{cd}-t$ as the TTM of contract $(c,d)$ at time $t\in\mathbb{N}$.

For $t\in\mathbb{N}$ and $c\in\mathcal{C}$, we let $S_{tc}>0$ denote the price of commodity $c$ at time $t$, with return $R_{tc} = S_{tc}/S_{t-1,c} - 1$. 
We assume that $S_{tc}$ is defined for all $t\in\mathbb{N}$ and $c\in\mathcal{C}$.  
For $t\in\mathbb{N}$ and $(c,d)\in\mathcal{C}\times\mathcal{D}_c$, let $V_{tcd} \geq 0$ denote the trading volume of futures contract $(c,d)$ at time $t$, and let $F_{tcd}>0$ denote its price at time $t$ for all $t \leq T_{cd}$. 
Assume that the margin ratio is one for every position (long or short) across all contracts and time, and that there are no transaction costs.
The return of futures contract $(c,d)$ is then defined as $R_{tcd}=F_{tcd}/F_{t-1,cd}-1$ when $t \leq T_{cd}$.  
We assume that all $S_{tc}$, $R_{tc}$, $V_{tcd}$, $F_{tcd}$, and $R_{tcd}$ are random variables, and that, for all $(c,d)\in\mathcal{C}\times \mathcal{D}_c$ and $t\leq T_{cd}$, 
\begin{equation}
 \mathbb{E} |R_{tcd}| < \infty,~
 \mathbb{E} R_{tcd}^2 < \infty,~
 \text{and}~
 \operatorname{Var}(R_{tcd})>0. 
 \label{eq:finite_return}
\end{equation}
We define the history at time $t$ as $\mathcal{H}_t=\{(t',\tilde H_{t'}): t' \in [1:t] \}$, where $\tilde H_{t'}=\{(c,d, \tilde V_{t'cd},  \tilde F_{t'cd}):(c,d)\in\mathcal{C}\times\mathcal{D}_c,  \tilde V_{t'cd}>0\}$ is an observable sample at $t'$.

The price of a futures contract $(c,d) \in \mathcal{C} \times \mathcal{D}_c$ at time $t$ is expressed as
\begin{equation}
    F_{tcd} = S_{tc}\exp (Q_{tcd}(T_{cd}-t)) >0 \label{eq:futures_price}
\end{equation}
where the random variable $Q_{tcd} \in \mathbb{R}$ is the net cost-of-carry rate. 
% , and $T_{cd}-t$ represents the TTM. 
Its return is then expressed as
\begin{equation} \label{eq:ret_futures}
   R_{tcd}=(1+R_{tc}) \exp(A_{tcd})-1   
\end{equation}
where $A_{tcd}=Q_{tcd}(T_{cd}-t)-Q_{t-1,cd}(T_{cd}-t+1)$ represents the change in the carry cost from $t-1$ to $t$. 
The spot delta of $(c,d)$ is
\begin{equation}
 \Delta_{tcd}=\frac{\partial F_{tcd}}{\partial S_{tc}}
=\exp(Q_{tcd}(T_{cd}-t))>0, \label{eq:delta}   
\end{equation}
which represents the risk exposure to the underlying commodity price \citep{hull2013fundamentals}.

We define the set of contracts observable from $t-1$ to $t$ as $U_t=\{(c,d)\in\mathcal{C}\times\mathcal{D}_c: V_{tcd} V_{t-1,cd}>0\}$ with $n_t=|U_t|$. 
Let $C_t=\{c:(c,d)\in U_t\}$ and $D_{tc}=\{d:(c,d)\in U_t\}$ denote the corresponding sets of commodities and contracts, respectively, with $n_{tc}=|D_{tc}|$. 
We denote the realization of $U_t$ by $\tilde{U}_t=\{(c,d)\in\mathcal{C}\times\mathcal{D}_c: \tilde V_{tcd} \tilde V_{t-1,cd}>0\}$, and define $\tilde C_t$, $\tilde D_{tc}$, $\tilde n_t$, and $\tilde n_{tc}$, analogously. 
We denote the trading universe at time $t$ by $\hat{U}_t \subset \{(c,d)\in \mathcal{C}\times\mathcal{D}_c: t \leq T_{cd}\}$, and define $\hat{C}_t, \hat{D}_{tc}, \hat{n}_t$, and $\hat{n}_{tc}$ analogously. 
Note that $\hat{U}_t$ is determined at $t-2$. 
In the following, we consider a single investor implementing our method.

\section{Proposed Method} 

To exploit statistical arbitrage in commodity futures markets, we propose a hierarchical graph learning method to predict futures price movements for CS trading. 
We first formally define a CS strategy, establish its efficacy through theoretical propositions, and introduce a projection method that maps predictions into CS positions.
We then formulate the prediction problem for CS trading and develop a hierarchical graph learning architecture that captures TTM-dependent interrelationships among commodity futures.

\subsection{Calendar Spread Strategy}
\label{subsec:cs}
For $t\in\mathbb{N}$ and $c\in \hat C_t$, we define a CS strategy for $c$ as a trading strategy that yields a position weight vector $\mathbf{w}_{tc}=[w_{tcd}]_{d\in \hat D_{tc}} \in\mathbb{R}^{\hat n_{tc}}$ such that $ \mathbf 1^\top \mathbf{w}_{tc}=0$ and $\|\mathbf{w}_{tc}\|_1=1$. 
This implies symmetric long and short positions: 
\begin{equation}
 \sum_{d\in \hat D_{tc}}[w_{tcd}]^+=\sum_{d\in \hat D_{tc}}[-w_{tcd}]^+=\frac{1}{2} .\label{eq:symmetric_weight}
\end{equation}
We let $\mathcal W_{tc,\mathrm{CS}}=\{\mathbf{w}_{tc}\in\mathbb R^{\hat n_{tc}}: \mathbf 1^\top \mathbf{w}_{tc}=0,\ \|\mathbf{w}_{tc}\|_1=1\}$ denote the set of CS weights for $c$ at $t$, and $\mathcal W_{tc,\mathrm{LO}}=\{\mathbf{w}_{tc}\in\mathbb R^{\hat n_{tc}}: w_{tcd} \ge 0,\ \|\mathbf{w}_{tc}\|_1=1\}$ denote the set of LO weights for $c$ at $t$. 
Since any position can be expressed as a linear combination of CS and LO positions, we focus on these two classes.

We evaluate both classes of strategies in terms of their risk and risk-adjusted return.  
For $t \in \mathbb N$ and $c \in \hat C_t$, a position weight vector $\mathbf{w}_{tc}=[w_{tcd}]_{d\in \hat D_{tc}} \in\mathbb{R}^{\hat n_{tc}}$ is determined by an investor at $t-2$, built in markets at $t-1$. 
Then, the resulting position is cleared at $t$, yielding the return given by 
\begin{equation}\label{eq:Gaam_def}
 \Pi(\mathbf{w}_{tc})=\sum_{d\in \hat D_{tc}} w_{tcd} R_{tcd} %=\sum_{d\in D_{tc}} w_{tcd} ((1+R_{tc})e^{A_{tcd}}-1).    
\end{equation}
Note that $\Pi(\mathbf{w}_{tc})$ depends on $R_{tcd}$; for simplicity, we omit this dependence. 
Moreover, we define the return variance of $\mathbf{w}_{tc}$ as $\operatorname{Var}(\Pi(\mathbf{w}_{tc}))$, the IR of $\mathbf{w}_{tc}$ as $\operatorname{IR}(\Pi(\mathbf{w}_{tc}))=\mathbb{E}(\Pi(\mathbf{w}_{tc}))/\sqrt{\operatorname{Var}(\Pi(\mathbf{w}_{tc}))}$ if $\operatorname{Var}(\Pi(\mathbf{w}_{tc}))>0$, 
and the delta of $\mathbf{w}_{tc}$ as
\begin{equation}
 \Delta(\mathbf{w}_{tc})=\tfrac{\partial}{\partial S_{t-1,c}}(\sum_{d\in \hat {D}_{tc}} w_{tcd} F_{t-1,cd}) \allowbreak =\sum_{d\in \hat D_{tc}} w_{tcd} \Delta_{t-1,cd}.   \label{eq:def_port_delta}
\end{equation}
Risk is commonly measured by return variance \citep{markowitz1952portfolio}  or by delta \citep{hull2013fundamentals}, and investment objectives are often framed as maximizing risk-adjusted return \citep{sharpe1966mutual,sharpe1998sharpe}. 
Thus, we analyze LO and CS strategies under these metrics in \Cref{prop:var,prop:ir,prop:low_delta}, with proofs in Appendix \Cref{appendix:proofs}.

\Cref{prop:var} implies that CS strategies can have lower variance than LO strategies, thereby entailing lower risk. 
Since the variance of $R_{tcd}$ is driven more by the underlying commodity return $R_{tc}$ than by the cost-of-carry change $A_{tcd}$ in \cref{eq:ret_futures} \citep{gospodinov2013commodity,maslyuk2009cointegration}, returns of contracts with the same underlying commodity tend to be positively correlated. 
Based on this property, we assume that such correlations admit nonnegative lower bounds in \Cref{prop:var}. 
Part (i) shows that an LO position has higher variance than all CS positions when correlations exceed a threshold depending on the LO position. 
Part (ii) removes this dependence and shows that all CS positions have lower variance than all LO positions when correlations exceed a bound determined by the variance range ratio and the number of contracts.

\begin{proposition} \label{prop:var} 
Fix $t\in\mathbb{N}$ and $c\in \hat C_t$. 
For $d,d'\in \hat D_{tc}$, let $\sigma_{tcdd'}=\operatorname{Cov}(R_{tcd},R_{tcd'})$, $\rho_{tcdd'}=\operatorname{Corr}(R_{tcd},R_{tcd'})=\sigma_{tcdd'}/\sqrt{\sigma_{tcdd}\sigma_{tcdd'}}$, and 
$\kappa_{tc}^2=(\max_{d\in \hat D_{tc}}\sigma_{tcdd})/(\min_{d\in \hat D_{tc}}\sigma_{tcdd})$. \\
\textnormal{(i)} If $\min_{d,d'\in \hat D_{tc}}\rho_{tcdd'} \geq 0$ and $\mathbf{w}\in\mathcal{W}_{tc,\mathrm{LO}}(\hat D_{tc})$ satisfies 
\begin{equation}
\min_{d,d'\in \hat D_{tc}}\rho_{tcdd'} > ({\kappa_{tc}^2-2\norm{\mathbf{w}}_2^2})({3-2\norm{\mathbf{w}}_2^2})^{-1},    \label{eq:prop-var-cond1}
\end{equation}
then for all $\mathbf{w}'\in \mathcal{W}_{tc,\mathrm{CS}}(\hat D_{tc})$, $\operatorname{Var}(\Pi(\mathbf{w})) > \operatorname{Var}(\Pi(\mathbf{w}'))$.\\
\textnormal{(ii)} 
If $\min_{d,d'\in \hat D_{tc}}\rho_{tcdd'} \geq 0$ and 
\begin{equation}
\min_{d,d'\in \hat D_{tc}}\rho_{tcdd'} > (\kappa_{tc}^2 \hat n_{tc} - 2)(3 \hat n_{tc}-2)^{-1},
\label{eq:prop-var-cond2}
\end{equation} 
then for all $(\mathbf{w},\mathbf{w}') \in \mathcal{W}_{tc,\mathrm{LO}}(\hat D_{tc})\times\mathcal{W}_{tc,\mathrm{CS}}(\hat D_{tc})$, $\operatorname{Var}(\Pi(\mathbf{w})) > \operatorname{Var}(\Pi(\mathbf{w}'))$. 
\end{proposition}

\Cref{prop:ir} shows that if the expected return of an LO strategy $\mathbf{w}$ exceeds that of the uninformed equal-weight strategy $\bar{\mathbf{w}}$ by a threshold, then the CS strategy $\check{\mathbf{w}}$ derived from $\mathbf{w}$ attains \red{a positive IR that exceeds that of the LO strategy.}  
The CS weight vector $\check{\mathbf{w}}$ is obtained by demeaning and rescaling $\mathbf{w}$, thereby preserving the relative ordering of weights while satisfying the CS constraints. 
Because $\operatorname{Var}(\Pi(\check{\mathbf{w}}))\ge 0$ in part (ii), the first inequality in each part ensures that $\Pi(\mathbf{w})$ has nonzero variance. 
Since $\mathbf{w}\geq 0$, this nonzero-variance condition typically holds when contracts with the same underlying commodity are positively correlated, as is commonly observed in the market (see \Cref{prop:var}). 
Moreover, the first inequality in part (ii) is satisfied under \Cref{prop:var}, for instance.  
The second inequality in parts (i) and (ii) excludes the degenerate case addressed in part (iii). 
Each part specifies a threshold on $\mathbb{E}\Pi(\mathbf{w})-\mathbb{E}\Pi(\bar{\mathbf{w}})$ under different conditions: the variance ratio in (i), the position difference in (ii), and zero in (iii). 
In part (iii), $\check{\mathbf{w}}$ yields a risk-free arbitrage, which can be viewed as having an infinite IR, and in parts (i) and (ii), the CS strategy achieves \red{a positive IR that exceeds that of the LO strategy.} 
Therefore, if an LO strategy exceeds the uninformed benchmark beyond the corresponding threshold, investing in the LO position is inferior to the derived CS position. 
Since CS positions are dollar-neutral, their Sharpe and information ratios coincide, whereas the Sharpe benchmark for LO positions is the risk-free rate \citep{sharpe1966mutual,sharpe1998sharpe}. 
Thus, \Cref{prop:ir} also implies that CS strategies can attain higher Sharpe ratios than LO strategies.

% Fix $t\in\mathbb{N}$,  $c\in\mathcal{C}$, and $D'_{tc} \subseteq D_{tc}$ with $n'_{tc}=|D'_{tc}|> 1$. 

\begin{proposition} \label{prop:ir}
Fix $t\in\mathbb{N}$ and $c\in \hat C_t$. 
Let $\mathbf{w}\in\mathcal{W}_{tc,\mathrm{LO}}$, $\bar{\mathbf{w}}=\mathbf 1/\hat n_{tc} \in\mathcal{W}_{tc,\mathrm{LO}}$, and $\check{\mathbf{w}}=(\mathbf{w}-\bar{\mathbf{w}})/\norm{\mathbf{w}-\bar{\mathbf{w}}}_1\in\mathcal{W}_{tc,\mathrm{CS}}$ with $\mathbf{w}\neq \bar{\mathbf{w}}$. 
\\
\textnormal{(i)} If $\operatorname{Var}(\Pi(\mathbf{w})) > 0$, $\operatorname{Var}(\Pi(\mathbf{w}-\bar{\mathbf{w}})) > 0$, and
\begin{equation}\label{eq:prop-ir-cond1}
\mathbb{E}\Pi(\mathbf{w})-\mathbb{E}\Pi(\bar{\mathbf{w}})
>
\red{[\mathbb{E}\Pi(\mathbf{w})]^+} \sqrt{\operatorname{Var}(\Pi(\mathbf{w}-\bar{\mathbf{w}}))/\operatorname{Var}(\Pi(\mathbf{w}))}, 
\end{equation} 
then $\operatorname{IR}(\Pi(\check{\mathbf{w}}))>\red{[\operatorname{IR}(\Pi(\mathbf{w}))]^+}$. \\
\textnormal{(ii)} 
If $\operatorname{Var}(\Pi(\mathbf{w}))
>\operatorname{Var}(\Pi(\check{\mathbf{w}}))$, 
$\operatorname{Var}(\Pi(\mathbf{w}-\bar{\mathbf{w}})) > 0$, 
and
\begin{equation} \label{eq:prop-ir-cond2}
\mathbb{E}\Pi(\mathbf{w})-\mathbb{E}\Pi(\bar{\mathbf{w}})
>
\norm{\mathbf{w}-\bar{\mathbf{w}}}_1 \left[\mathbb{E}\Pi(\mathbf{w})\right]^+,  
\end{equation}
then $\operatorname{IR}(\Pi(\check{\mathbf{w}}))>\red{[\operatorname{IR}(\Pi(\mathbf{w}))]^+}$. \\
\textnormal{(iii)} If $\operatorname{Var}(\Pi(\mathbf{w})) > 0$, $\operatorname{Var}(\Pi(\mathbf{w}-\bar{\mathbf{w}})) = 0$, and
\begin{equation} \label{eq:prop-ir-cond0}
\mathbb{E}\Pi(\mathbf{w})-\mathbb{E}\Pi(\bar{\mathbf{w}}) > 0,
\end{equation}
then $\mathbb{E}\Pi(\check{\mathbf w}) > 0$, $\operatorname{Var}(\Pi(\check{\mathbf w}))=0$ and $\operatorname{IR}(\Pi(\mathbf w))<\infty$. 
\end{proposition}

\Cref{prop:low_delta} implies that CS strategies can have lower risk exposure to the underlying commodity price. 
Part (i) gives a condition under which the CS weight vector $\check{\mathbf{w}}$, obtained by demeaning an LO vector $\mathbf{w}$ and scaling, has lower delta. 
In particular, if $\norm{\mathbf{w}-\bar{\mathbf{w}}}_1 \geq 1$, then $\check{\mathbf{w}}$ always has lower risk exposure regardless of $\Delta(\mathbf{w})$ and $\Delta(\bar{\mathbf{w}})$, as both are positive by \cref{eq:delta}. 
Under the assumption that the delta range ratio, $\Delta_{t-1,c}^{\max}/\Delta_{t-1,c}^{\min}$, is less than three, 
part (ii) implies that, almost surely, every CS strategy is less exposed to the underlying commodity price than every LO strategy.

\begin{proposition}\label{prop:low_delta}
Fix $t\in\mathbb{N}$ and $c \in \hat C$.  \\
\textnormal{(i)} 
Let $\mathbf{w}\in\mathcal{W}_{tc,\mathrm{LO}}$, 
$\bar{\mathbf{w}}=\mathbf 1/\hat n_{tc} \in\mathcal{W}_{tc,\mathrm{LO}}$, and $\check{\mathbf{w}}=(\mathbf{w}-\bar{\mathbf{w}})/\norm{\mathbf{w}-\bar{\mathbf{w}}}_1\in\mathcal{W}_{tc,\mathrm{CS}}$ with $\mathbf{w}\neq \bar{\mathbf{w}}$.  
If $\Delta(\bar{\mathbf{w}})>(1-\norm{\mathbf{w}-\bar{\mathbf{w}}}_1)\Delta(\mathbf{w})$, then $\Delta(\check{\mathbf{w}})<\Delta(\mathbf{w})$. \\
\textnormal{(ii)} Assume that $\Delta_{t-1,cd}$ is an a.s. finite random variable for every~$d\in \hat D_{tc}$.
Let $\Delta_{t-1,c}^{\min}=\min\{\operatorname*{ess\,inf}\Delta_{t-1,cd}: d\in \hat D_{tc}\}\in\mathbb{R}$ and
$\Delta_{t-1,c}^{\max}=\max\{\operatorname*{ess\,sup}\Delta_{t-1,cd}: d\in \hat D_{tc}\}\in\mathbb{R}$.
If $3\Delta_{t-1,c}^{\min}>\Delta_{t-1,c}^{\max}$, then 
\begin{equation}
 \operatorname*{ess\,inf}_{\mathbf{w}\in\mathcal{W}_{tc,\mathrm{LO}}}\Delta(\mathbf{w})
 >
 \operatorname*{ess\,sup}_{\mathbf{w}\in\mathcal{W}_{\mathrm{tc,CS}}}\Delta(\mathbf{w})
\end{equation}
\end{proposition}

We assess the realism of the inequality assumptions in the second parts of \Cref{prop:var,prop:low_delta} by estimating the relevant parameters in the trading universe $\hat U_t$. 
Since $\hat U_t$ is determined based on $\mathcal{H}_{t-2}$, as discussed below in \cref{eq:informative_features}, some contracts in $\hat U_t$ may not be observable at $t-1$ or $t$, i.e., $\tilde V_{t-1,cd} \tilde V_{tcd}=0$. 
Thus, in the propositions, we replace $\hat D_{tc}$, $\hat n_{tc}$, and $\hat C_t$    with $\hat D'_{tc}=\hat D_{tc} \cap \tilde D_{tc}$, $\hat n'_{tc}=|\hat D'_{tc}|$, and $\hat C'_t=\{c\in \hat C_t \cap \tilde C_t: \hat n'_{tc} >1 \}$, respectively.  
For a CS strategy to be implemented, at least two contracts on the same underlying asset are required; thus, we impose the constraint $\hat n'_{tc} > 1$ in $\hat C'_t$. 
For each $t\in\mathbb{N}$, $c\in\hat C'_{t}$, and $(d,d')\in  \hat D'_{t,c}\times \hat D'_{t,c}$, 
we estimate $\hat{\sigma}_{tcdd'}$ as the sample covariance of $\{(\tilde{R}_{t'cd},\tilde{R}_{t'cd'}): t'\in [t-n_\mathrm{sam}^\mathrm{min}:t]\}$, assuming $\hat{\sigma}_{tcdd'}\approx\hat{\sigma}_{t'cdd'}$ over $t'\in [t-n_\mathrm{sam}^\mathrm{min}:t)$. 
Since $\hat{U}_{t}$ and $\tilde U_t$ ensure price data availability on $[t-n_\mathrm{sam}^\mathrm{min}-1:t-2]$---discussed in \cref{eq:univ_past_trades}---and $[t-1:t]$, respectively, we use returns over $[t-n_\mathrm{sam}^\mathrm{min}:t]$ for estimation. 
We then estimate 
$\hat\rho_{tcdd'}=\hat\sigma_{tcdd'}/\sqrt{\hat \sigma_{tcdd} \hat\sigma_{tcdd'}}$, $\min_{d,d'\in \hat D'_{tc}}\hat \rho_{tcdd'}$, and 
$\hat \kappa_{tc}^2=(\max_{d\in \hat D'_{tc}} \hat \sigma_{tcdd})/(\min_{d\in \hat D'_{tc}}\hat \sigma_{tcdd})$. 
For each $t\in\mathbb{N}$ and $c\in\hat C_{t}$, 
we estimate $\hat\Delta_{t-1,c}^\mathrm{max}/\hat\Delta_{t-1,c}^\mathrm{min}=\max\{\tilde{F}_{t-1,cd}/\tilde{F}_{t-1,cd'}:d,d' \in \hat D'_{tc}\}$ from \cref{eq:futures_price,eq:delta}. 
We then check the following inequalities for each $t$ and $c$: 
\begin{align}
&\hat{\rho}_{tc,\min}-(\hat{\kappa}^2 \hat{n}'_{tc} - 2)(3 \hat{n}'_{tc}-2)^{-1}>0, \label{eq:test1}\\
&\hat{\rho}_{tc,\min}\geq 0, \label{eq:test2}\\
&\hat{\Delta}_{tc}^{\max}/\hat{\Delta}_{tc}^{\min} < 3, \label{eq:test3}
\end{align}
where \cref{eq:test1,eq:test2} correspond to (ii) in \Cref{prop:var},  and \cref{eq:test3} corresponds to $3\Delta_{t-1,c}^{\min}>\Delta_{t-1,c}^{\max}$ in \Cref{prop:low_delta}.

Since CS strategies can exhibit lower risk and higher IR yet remain largely unexplored in the machine-learning literature, we focus on CS strategies and extend the framework from a single commodity to multiple commodities.
We define a CS strategy for $\hat U_t$ as a strategy that yields a position weight vector $\mathbf{w}_t=[w_{tcd}]_{(c,d)\in \hat U_{t}} \in \mathbb{R}^{\hat n_t}$ satisfying $\|\mathbf{w}_{t}\|_1=1$ and 
\begin{equation}
  \mathbf 1^\top \mathbf{w}_{tc}=0,~\forall c\in \hat C_t.  
  \label{eq:def_cs_multi}
\end{equation}
Unlike the single-commodity definition, the $\ell_1$-norm constraint is imposed on $\mathbf{w}_t$, rather than on each $\mathbf{w}_{tc}$. 
A CS strategy for $\hat U_t$ is a convex combination of the weight vectors of CS strategies, one for each $c \in \hat C_t$, after extending their dimensions to $\mathbb{R}^{\hat n_t}$. 
 
% \red{$\hat{\mathbf{Y}}'_{t} \leftarrow \hat{\mathbf Y}^\circ$}

% \red{$\hat{\mathbf Y}^\circ \leftarrow  \hat{\mathbf Y}^\circ_t$}

% \red{$\bar{\hat{Y}}_{tc} \leftarrow \hat{\mathbf Y}^\bullet$}

% \red{$\hat{\mathbf Y}^\bullet \leftarrow \hat{Y}^\bullet$}

% \red{$\hat{Y}^\bullet \leftarrow \hat{Y}^\bullet_{tc}$}

Given a prediction vector $\hat{\mathbf{Y}}_{t}=[\hat Y_{tcd}]_{(c,d)\in\hat U_t} \in \mathbb{R}^{\hat{n}_t}$, we construct a CS weight vector $\mathbf{w}_t$ for $\hat U_t$ using a projection method adapted from \cite{hong2025statistical}, where $\hat Y_{tcd}$ is discussed below.
Specifically, we compute 
\begin{equation}
\mathbf{w}_t = \hat{\mathbf Y}^\circ_t/\|\hat{\mathbf Y}^\circ_t\|_1    \label{eq:group-neutralized-position-vector1}
\end{equation}
where 
\begin{align}
 \hat{\mathbf Y}^\circ_t&=[\hat{Y}_{tcd}-\hat{Y}^\bullet_{tc}]_{(c,d)\in\hat{U}_t},    \label{eq:group-neutralized-position-vector2}\\
 \hat{Y}^\bullet_{tc}&=\frac{1}{\hat{n}_{tc}} \sum_{d' \in \hat{D}_{tc}}  \hat{Y}_{tcd'}.  \label{eq:group-neutralized-position-vector3}
\end{align}
That is, we center the prediction vector by subtracting the commodity-wise average $\hat{Y}^\bullet_{tc}$ in \cref{eq:group-neutralized-position-vector2}, and then scale it in \cref{eq:group-neutralized-position-vector1}.
This computation is equivalent to projecting $\hat{\mathbf{Y}}_{t}$ onto $\{\mathbf{w}_t:\mathbf 1^\top \mathbf{w}_{tc}=0,~\forall c\in \hat{C}_t\}$ under the Euclidean norm and then scaling; see Appendix \Cref{appendix:group_neutralization}. 
Thus, $\mathbf{w}_t$ satisfies $\mathbf 1^\top \mathbf{w}_{tc}=0,~\forall c\in \hat{C}_t$ and $\|\mathbf{w}_{t}\|_1=1$. 
Since projection yields the closest vector satisfying the constraints, the resulting vector from \Cref{eq:group-neutralized-position-vector1,eq:group-neutralized-position-vector2,eq:group-neutralized-position-vector3} is, up to scaling, the CS weight vector closest to the prediction vector.

We consider a daily trading scheme with daily predictions and periodic retraining. 
Given $n_\mathrm{fit}\in\mathbb{N}$, we let $\{t_k: k\in [1:{n_\mathrm{fit}}] \}$ denote the retraining dates with $t_k<t_{k+1}$ and $t_{n_\mathrm{fit}+1}=\infty$.  
For each $k\in[1:n_\mathrm{fit}]$, a prediction model $f_{k}$ is trained on $\mathcal{H}_{t_k}$ when $\mathcal{H}_{t_k}$ becomes available, i.e., after the markets close on $t_k$ and before they open on $t_k+1$.  
On each trading date $t\in \mathcal{T}_k^\mathrm{test}=[t_k:t_{k+1})$, 
we predict price movements $\mathbf{Y}_{t+2}=[Y_{t+2,cd}]_{(c,d)\in \hat U_{t+2}}\in\mathbb{R}^{\hat n_{t+2}}$ from $t+1$ to $t+2$ as $\hat{\mathbf{Y}}_{t+2}=[\hat{Y}_{t+2,cd}]_{(c,d) \in \hat{U}_{t+2}}=f_{k}(\mathcal{H}_t)\in\mathbb{R}^{\hat{n}_{t+2}}$ where $Y_{t+2,cd}$ is defined below.
Based on these predictions, CS position vector 
$\mathbf{w}_{t+2}=[w_{t+2,cd}]_{(c,d)\in\hat{U}_{t+2}}\in\mathbb{R}^{\hat{n}_{t+2}}$ is computed via \cref{eq:group-neutralized-position-vector1,eq:group-neutralized-position-vector2,eq:group-neutralized-position-vector3}. 
The prediction and CS position vectors are generated after the markets close at $t$ (or after training if $t=t_k$) and before they open at $t+1$. 
In the markets, the CS position vector is then established at $t+1$ and liquidated at $t+2$, with trades executed at prices $F_{t+1,cd}$ and $F_{t+2,cd}$, respectively.
No datum in $H_{t'}$ with $t'\ge t_k+1$ is used to train $f_k$, and no datum in $H_{t'}$ with $t'\ge t+1$ is used to produce $(\hat{\mathbf{Y}}_{t+2},\mathbf{w}_{t+2})$, thereby precluding look-ahead bias.

\subsection{Prediction Model}
\label{subsec:prediction}

% \red{$\check R_{tcd} \leftarrow R_{tcd}^\circ$}

% \red{$\bar R_{tcd} \leftarrow R_{tcd}^\bullet $}

For each $k\in[1:n_\mathrm{fit}]$, given history $\mathcal{H}_{t_k}$, our prediction problem is to train a model $f_{k}(\mathcal{H}_t)=[\hat{Y}_{t+2,cd}]_{(c,d) \in \hat{U}_{t+2}}\in\mathbb{R}^{\hat{n}_{t+2}}$ that predicts $[Y_{t+2,cd}]_{(c,d)\in {\hat U_{t+2}}}$, so as to minimize
\begin{equation}
\sum_{t\in \mathcal{T}_k^\mathrm{test}}\mathbb{E} \left[ l_{p} \left(
[\hat{Y}_{t+2,cd}]_{(c,d)\in {\hat{U}_{t+2}}}, 
[Y_{t+2,cd}]_{(c,d)\in {\hat{U}_{t+2}}} 
\right) \right],  
\end{equation}
where $\mathcal{T}_k^\mathrm{test}=[t_k:t_{k+1})$ and $l_p$ is a loss function. 
For $t\in\mathbb{N}$ and $(c,d)\in \hat U_t$, the target $Y_{tcd}$ is defined as the normalized, commodity-wise demeaned return
% \red{$\check R_{tcd} \leftarrow R_{tcd}^\circ$}
% \red{$\bar R_{tcd} \leftarrow R_{tcd}^\bullet $}
% \begin{equation}
%     Y_{tcd}=\Phi^{-1}(\widehat{\operatorname{cdf}}(\check{R}_{tcd};[\check{R}_{tcd}]_{(c,d)\in \hat{U}_{t}}))  \label{eq:prob_int}
% \end{equation} 
\begin{equation}
    Y_{tcd}=\Phi^{-1}(\widehat{\operatorname{cdf}}(R^\circ_{tcd};[R^\circ_{tcd}]_{(c,d)\in \hat{U}_{t}}))  \label{eq:prob_int}
\end{equation} 
where $\Phi^{-1}$ is the standard normal quantile function, and the commodity-wise demeaned return is defined as
\begin{equation}
 R^\circ_{tcd}=R_{tcd}-R^\bullet_{tc}   \label{eq:demeaned_return}
\end{equation}
with $R^\bullet_{tc}=\mathrm{avg}(R_{tcd}; d \in  \hat{D}_{tc})$. 
By the probability integral transform, $\widehat{\operatorname{cdf}}(R^\circ_{tcd};[R^\circ_{tcd}]_{(c,d)\in  \hat{U}_{t}})$ is approximately uniformly distributed on $[0,1]$, and $\Phi^{-1}$ maps it to an approximately standard normal variable. 
Since the standard-normal target can stabilize optimization \citep{lecun2002efficient,glorot2010understanding,hong2023index}, we adopt this prediction target. 
The motivation for using $R^\circ_{tcd}$ is discussed below.

When addressing a task related to $\Pi(\mathbf{w}_t)$ for a CS strategy, we can predict $R^\circ_{tcd}$ in lieu of $R_{tcd}$.   
For $\mathbf{w}_t\in\mathcal{W}_{t,\mathrm{CS}}$, its strategy return satisfies 
$\Pi(\mathbf{w}_t)
=\sum_{c\in \hat C_t} \sum_{d\in \hat D_{tc}} w_{tcd} R_{tcd}
=\sum_{c\in \hat C_t} (\sum_{d\in \hat D_{tc}}   w_{tcd} R_{tcd} -  R^\bullet_{tc}\sum_{d\in \hat D_{tc}} w_{tcd})
=\sum_{c\in \hat C_t} \allowbreak \sum_{d\in \hat D_{tc}}  w_{tcd} (R_{tcd}-R^\bullet_{tc})
=\sum_{(c,d)\in \hat U_t} R^\circ_{tcd} w_{tcd}$ 
where each equality follows from \cref{eq:Gaam_def}, \cref{eq:def_cs_multi}, algebraic manipulation, and \cref{eq:demeaned_return}, in that order. 
Thus, $R^\circ_{tcd}$ can replace $R_{tcd}$ in \cref{eq:Gaam_def} to obtain the return of a CS strategy.

% Due to their different exposures to the underlying commodity, 
As $R^\circ_{tcd}$ can be less exposed to the underlying commodity, 
$R^\circ_{tcd}$ can be more predictable than $R_{tcd}$. 
Fix $(c,d)\in \hat U_t$. 
The quantity $R^\circ_{tcd}$ can be interpreted as the return of a CS strategy for $c$ with a position vector $[\mathbb{I}_{\{d'=d\}}]_{d'\in \hat D_{tc}}-\tfrac{1}{\hat n_{tc}} \mathbf 1 \in \mathbb R^{\hat n_{tc}}$, which becomes a weight vector after scaling.
In contrast, $R_{tcd}$ corresponds to the return of a LO strategy for $c$ with weight $[\mathbb{I}_{\{d'=d\}}]_{d'\in \hat D_{tc}}$. 
By \Cref{prop:low_delta} and its experimental evidence in \Cref{sec:experiments}, $R_{tcd}^\circ$ has lower exposure to the underlying commodity than $R_{tcd}$ in a large fraction of cases. 
In addition, the data for underlying commodities are intrinsically incomplete.\footnotemark\footnotetext{\label{fn:proxy}
True spot markets for some underlying commodities do not exist \citep{gibson1990stochastic}, and thus their spot price data are not directly observable.} 
Therefore, we predict $R_{tcd}^\circ$.

For each $k\in[1:n_\mathrm{fit}]$, 
we train $f_k$ on
$\{(\mathcal{H}_t,[\tilde Y_{t+2,cd}]_{(c,d)\in \hat U_{t+2} \cap \tilde U_{t+2}}):t\in \mathcal{T}_k^\mathrm{fit}\}$,
validate on
$\{(\mathcal{H}_t,[\tilde Y_{t+2,cd}]_{(c,d)\in \hat U_{t+2} \cap \tilde U_{t+2}}):t\in \mathcal{T}_k^\mathrm{val}\}$,
and evaluate on
$\{(\mathcal{H}_t,[\tilde Y_{t+2,cd}]_{(c,d)\in \hat U_{t+2} \cap \tilde U_{t+2}}):t\in \mathcal{T}_k^\mathrm{test}\}$,
where $\mathcal{T}_k^\mathrm{fit} \cap \mathcal{T}_k^\mathrm{val} = \emptyset$, $\mathcal{T}_k^\mathrm{fit} \cup \mathcal{T}_k^\mathrm{val} \subset [1:t_k-2]$, and $\mathcal{T}_k^\mathrm{test}=[t_k:t_{k+1})$. 
Since $\hat U_{t+2}$ is determined at $t$ based solely on $\mathcal{H}_{t}$, some contracts in $\hat U_{t+2}$ may not be observable at $t+1$ or $t+2$, i.e., $\tilde V_{t+1,cd} \tilde V_{t+2,cd}=0$; thus, we use $\hat U_{t+2} \cap \tilde U_{t+2}$ in lieu of $\hat U_{t+2}$ for the model training and evaluation:  
\begin{equation}
\tilde Y_{tcd}=\Phi^{-1}(\widehat{\operatorname{cdf}}(\tilde R^\circ_{tcd};[\tilde R^\circ_{tcd}]_{(c,d)\in \hat{U}_{t} \cap \tilde U_t})), \quad  
\tilde R^\circ_{tcd}=\tilde R_{tcd}-\tilde R^\bullet_{tc}, \quad
\tilde R^\bullet_{tc}=\mathrm{avg}(\tilde R_{tcd}; d \in  \hat{D}_{tc} \cap \tilde{D}_{tc}). 
\label{eq:empirical_target}
\end{equation}
Note that the training and validation datasets are constructed solely from $\mathcal{H}_{t_k}$, which is given for training $f_k$ in the problem, and do not use any future data, i.e., $\mathcal{H}_{t'}$ for $t'>t_k$, thereby ensuring no look-ahead bias.\footnote{The pair $(\mathcal{H}_{t_k-2},[\tilde Y_{t_k,cd}]_{(c,d)\in \hat U_{t_k} \cap \tilde U_{t_k}})$ is the latest sample that can be used for training or validation. 
By \cref{eq:prob_int,eq:demeaned_return} and the definition of $U_{t_k}$, 
the latest historical data involved in $(\mathcal{H}_{t_k-2},[\tilde Y_{t_k,cd}]_{(c,d)\in \hat U_{t_k} \cap \tilde U_{t_k}})$ are $\mathcal{H}_{t_k}$, thereby ensuring no look-ahead bias.}  
Additionally, $\mathcal{T}_k^\mathrm{test}$ for $k\in[1:n_\mathrm{fit}]$ partition $[t_1:t_{n_\mathrm{fit}+1})$, i.e.,
$\bigcup_{k\in[1:n_\mathrm{fit}]} \mathcal{T}_k^\mathrm{test}=[t_1:t_{n_\mathrm{fit}+1})$
with $\mathcal{T}_k^\mathrm{test}\cap \mathcal{T}_{k'}^\mathrm{test}=\emptyset$ for $k\neq k'$.

\textbf{Architecture Overview}. 
To solve the problem, we propose a bi-level hierarchical graph learning architecture $f_k (\mathcal{H}_t)=
(f_{k}^\mathrm{H} \circ f_{k}^\mathrm{B} \circ f_{k}^\mathrm{L} \circ f_{k}^\mathrm{P})(\mathcal{H}_t)
=[\hat{Y}_{t+2,cd}]_{(c,d)\in \hat U_{t+2}}$. 
The preprocessing function $f_{k}^\mathrm{P}$ maps $\mathcal{H}_t$ to a bi-level hierarchical graph, and the graph lifting function $f_{k}^\mathrm{L}$ aligns TTM grids across commodities by extending the graph along the TTM dimension. 
The bi-level convolution function $f_{k}^\mathrm{B}$ then performs TTM-conditioned propagation through the extended bi-level graph, followed by the head $f_{k}^\mathrm{H}$, which outputs predictions $[\hat{Y}_{t+2,cd}]_{(c,d)\in \hat U_{t+2}}$. 
Throughout these functions, our architecture captures TTM-dependent information embedded in interrelationships among commodity futures contracts.

\begin{figure}
    \centering
    \includegraphics[width=0.95\textwidth,keepaspectratio]{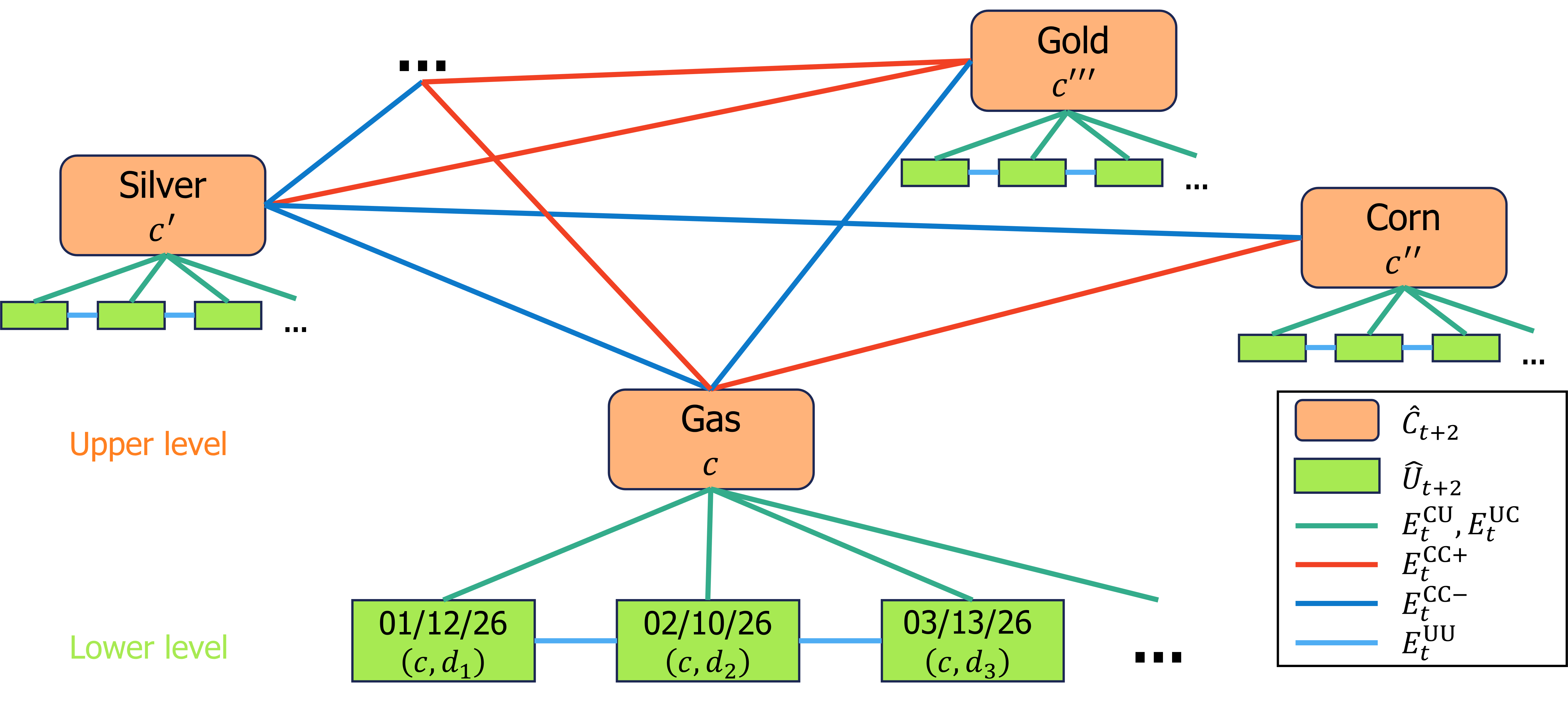}
    \captionof{figure}{Bi-Level Hierarchical Graph}
    \label{fig:hierarchical_graph}
\end{figure}

\textbf{The preprocessing function} $f_{k}^\mathrm{P}(\mathcal{H}_t)$ outputs a bi-level graph $G_{t}=(\hat{C}_{t+2} \cup \hat{U}_{t+2},  E_{t}^\mathrm{CC+} \cup E_{t}^\mathrm{CC-} \cup E_{t}^\mathrm{CU} \cup E_{t}^\mathrm{UC} \cup E_{t}^\mathrm{UU}, X_{t})$, illustrated in \Cref{fig:hierarchical_graph}. 
The commodities in $\hat{C}_{t+2}$ form the upper-level nodes, and their contracts in $\hat{U}_{t+2}$ constitute the lower-level nodes.
The upper-level nodes are connected via $E_{t}^{\mathrm{CC+}}\subset\hat{C}_{t+2}\times\hat{C}_{t+2}$ and $E_{t}^{\mathrm{CC-}}\subset\hat{C}_{t+2}\times\hat{C}_{t+2}$, which denote the edge sets corresponding to positive and negative correlations between commodities, respectively.
The cross-level connections are given by $E_{t}^{\mathrm{CU}}\subset \hat{C}_{t+2} \times \hat{U}_{t+2}$ and $E_{t}^{\mathrm{UC}}\subset \hat{U}_{t+2} \times \hat{C}_{t+2}$, linking each underlying commodity and its contracts. 
The lower-level nodes are connected via $E_{t}^\mathrm{UU} \subset \hat{U}_{t+2} \times \hat{U}_{t+2}$, which links each contract to its two nearest-maturity contracts (one shorter and one longer) with the same underlying commodity\footnote{For each commodity, contracts with the smallest or largest maturity have only one edge in $E_{t}^{\mathrm{UU}}$}, reflecting the economic insight \citep{brennan1976supply,gibson1990stochastic,schwartz1997stochastic,casassus2005stochastic} that futures contracts with similar TTMs share similar economic
characteristics.  
The matrix $X_{t} \in \mathbb{R}^{\hat n_{t+2} \times n_\mathrm{sam}^\mathrm{min}}$ contains node features for $\hat{U}_{t+2}$.

To ensure proper backtesting, we define the trading universe $\hat{U}_{t+2}$ as
\begin{equation}
    \hat{U}_{t+2} = \{(c,d)\in \hat{U}'_{t+2}: \hat \sigma_{tc}^x>0\}  \label{eq:informative_features} 
\end{equation}
where $\sigma_{tc}^x$ is defined below, and $\hat{U}'_{t+2}$ is defined as the set of $(c,d)\in \mathcal{C}\times\mathcal{D}_c$ such that
\begin{align}
    & t+2 \leq T_{cd}, \label{eq:univ_maturity}
    \\& T_{cd}-t \leq \tau^\mathrm{max}, \label{eq:univ_ttm}
    \\& \prod_{t'\in (t-n_\mathrm{sam}^\mathrm{min},t]} \tilde V_{t'cd} >0. \label{eq:univ_past_trades}
\end{align} 
Here, $\tau^\mathrm{max}\in\mathbb{N}$ and $n_\mathrm{sam}^\mathrm{min} \in\mathbb{N}$ are hyperparameters. 
The constraint in \cref{eq:informative_features} ensures $X_t$ is well-defined, as discussed below. 
\Cref{eq:univ_maturity} ensures that position vector determined at $t$ can be liquidated at $t+2$, and \cref{eq:univ_ttm} excludes long-TTM contracts, which tend to be illiquid \citep{xu2023neural}. 
\Cref{eq:univ_past_trades} implies that $(c,d)\in \hat{U}'_{t+2}$ is traded in the market on all of the most recent $n_\mathrm{sam}^\mathrm{min}$ trading dates, thereby ensuring liquidity. 
Unlike \cite{hu2025graph,tan2024futures}, our trading universe evolves over time, allowing node creation and deletion.

To enable neural networks to capture information in $\mathcal{H}_t$, we define the node feature matrix $X_{t} =[\mathbf{x}_{tcd}]_{(c,d)\in \hat U_{t+2}} \in \mathbb{R}^{\hat n_{t+2} \times n_\mathrm{sam}^\mathrm{min}}$ where $\mathbf{x}_{tcd}=[x_{tcd\tau}]_{\tau\in[0:n_\mathrm{sam}^\mathrm{min})}\in\mathbb{R}^{n_\mathrm{sam}^\mathrm{min}}$ and 
\begin{equation}
x_{tcd\tau} = \Phi^{-1} (\widehat{\operatorname{cdf}}( x'_{tcd\tau};[x'_{tcd\tau'}]_{(c,d)\in \hat U_{t+2},\tau' \in[0:n_\mathrm{sam}^\mathrm{min})})) \label{eq:normal-features} 
\end{equation}
for $t\in [n_\mathrm{sam}^\mathrm{min}:\infty)$, $(c,d) \in \hat U_{t+2}$, and $\tau\in[0:n_\mathrm{sam}^\mathrm{min})$. 
For $t\in [n_\mathrm{sam}^\mathrm{min}:\infty)$, $(c,d) \in \hat{U}'_{t+2}$, and $\tau\in[0:n_\mathrm{sam}^\mathrm{min})$, we define 
\begin{equation}
x'_{tcd\tau} = x''_{tcd\tau}/\hat \sigma_{tc}^x \label{eq:standardize-features} 
\end{equation}
where 
\begin{align}
\hat \sigma_{tc}^x &= \mathrm{std}(x''_{tcd''\tau};\tau\in[0:n_\mathrm{sam}^\mathrm{min}), (c,d'') \in \hat{U}'_{t+2}), \label{eq:std-features}
\\ x''_{tcd\tau} &= \log(\tilde F_{t-\tau,cd}) - \hat \mu_{t c d \cdot}^x - \hat \mu_{t c \cdot \tau}^x + \hat \mu_{t c \cdot \cdot}^x, \label{eq:cenering-features} 
\\ \hat \mu_{t c d \cdot}^x &= \mathrm{avg}(\log(\tilde F_{t-\tau',cd}) ;\tau'\in[0:n_\mathrm{sam}^\mathrm{min})), \label{eq:mu-features1}
\\  \hat \mu_{t c \cdot \tau}^x &= \mathrm{avg}(\log(\tilde F_{t-\tau,cd''}) ; d'' \in \{d': (c,d')\in \hat{U}'_{t+2}\}), 
\label{eq:mu-features2}
\\  \hat \mu_{t c \cdot \cdot}^x &= \mathrm{avg}(\log(\tilde F_{t-\tau',cd''}) ; \tau' \in[0:n_\mathrm{sam}^\mathrm{min}), 
d'' \in \{d': (c,d')\in \hat{U}'_{t+2}\}). \label{eq:mu-features3}
\end{align}
\Cref{eq:cenering-features,eq:mu-features1,eq:mu-features2,eq:mu-features3} perform two-way ANOVA-style centering of log prices. 
To account for differing price scales across commodities, \cref{eq:standardize-features} scales the centered values within each commodity. 
Thus, each node feature contains centered and scaled log-price information over the most recent $n_\mathrm{sam}^\mathrm{min}$ days. 
Since standard-normal features can stabilize optimization \citep{lecun2002efficient,glorot2010understanding,hong2023index}, we normalize $x'_{tcd\tau}$ in \cref{eq:normal-features}. 

% \red{$\check{R}_{tc} \leftarrow \hat{R}^\circ_{tc}$}

% \red{$\check{R}_{tcd}^* \leftarrow \tilde{R}_{tcd}^*$}

To capture inter-commodity relationships, we define the inter-commodity edge sets $E_t^\mathrm{CC+}=\{(c,c')\in\hat{C}_{t+2}\times\hat{C}_{t+2}: \hat{\rho}_{tcc'} \geq \rho^*\}$
and
$E_t^\mathrm{CC-}=\{(c,c')\in\hat{C}_{t+2}\times\hat{C}_{t+2}: \hat{\rho}_{tcc'} \leq -\rho^*\}$.
% with $E_t^\mathrm{CC-}$ defined analogously with the threshold $\leq -\rho^*$. 
Here, $\rho^*>0$ is a predefined threshold, and $\hat{\rho}_{tcc'}$ is the Pearson correlation computed from 
the accumulated sample set $\bigcup_{t' \in [1:t]} B_{t'cc'}$, provided that $|\{t'\in [1:t]: B_{t'cc'} \neq\emptyset\}| \geq n_\mathrm{sam}^\mathrm{min}$. 
If the minimum number of observation time points is not met, we set $\hat{\rho}_{tcc'}=0$. 
A sample batch $B_{tcc'}$ is defined as the set of maturity-aligned pairs of estimated-and-demeaned returns: 
$B_{tcc'}=\{(\hat{R}^\circ_{tc}(T), \hat{R}^\circ_{tc'}(T)): T \in (\mathcal{T}_c \cup \mathcal{T}_{c'}) \cap \mathrm{dom}(\hat{R}^\circ_{tc}) \cap \mathrm{dom}(\hat{R}^\circ_{tc'}), T-t\leq\tau^\mathrm{max} \}$.  
Here, $\hat{R}^\circ_{tc}(T)$ is the demeaned return estimation function with domain $\mathrm{dom}(\hat{R}^\circ_{tc})$. 
Since maturities are not aligned across commodities---i.e., a commodity-$c$ futures contract with maturity $T$, or its return, may not exist---we linearly interpolate $\hat{R}^\circ_{tc}(T)=\frac{\tilde{R}_{tcd^+}^*-\tilde{R}_{tcd^-}^*}{T_{cd^+}-T_{cd^-}}(T-T_{cd^-})+\tilde{R}_{tcd^-}^*$ 
where $(d^-,d^+)$ minimizes $|d^+-d^-|$ subject to $T \in [T_{cd^-}, T_{cd^+}]$ and both $\tilde{R}_{tcd^+}^*$ and $\tilde{R}_{tcd^-}^*$ are well-defined. 
We define 
\begin{equation}
\tilde{R}_{tcd}^*=\tilde{R}_{tcd}-\mathrm{avg} (\tilde{R}_{tcd'};d'\in D'_{tc} )    
\label{eq:demeaned_return_for_graph} 
\end{equation}
for $d\in D_{tc}$ with $|D'_{tc}|>1$, where $D'_{tc}=\{d' \in D_{tc}: T_{cd'}-t \leq \tau^\mathrm{max} \}$.
Consistent with the prediction target \cref{eq:prob_int}, which is based on demeaned returns \cref{eq:demeaned_return}, we use demeaned returns for $\hat{R}^\circ_{tc}(T)$.

\textbf{The graph lifting function} $f_k^\mathrm{L}(G_t)$ outputs an augmented graph $G_{t0}=(\hat{C}_{t+2,0} \cup \hat{U}_{t+2},  E_{t0}^\mathrm{CC+} \cup E_{t0}^\mathrm{CC-} \cup E_{t0}^\mathrm{CU} \cup E_{t0}^\mathrm{UC}\cup E_{t}^\mathrm{UU}, Z_{t0})$.
It transforms $\hat{C}_{t+2}$ into an extended node set $\hat{C}_{t+2,0}=\{(c,j):j \in [0:n_\mathrm{bas}]\}$ of virtual futures contracts where $j$ indexes a virtual basis TTM $\tau_j=j\cdot(\tau^\mathrm{max}/n_\mathrm{bas})$. 
% Using $\hat{C}_{t+2,0}$, we 
This aligns maturity grids $\mathcal{T}_c$ across commodities $c\in\hat C_{t+2}$ onto a shared virtual TTM grid $\{\tau_j:j\in [0:n_\mathrm{bas}]\}$, 
since commodities generally do not share identical futures maturities, i.e., $\mathcal{T}_{c'} \neq \mathcal{T}_{c''}$ for $c',c''\in\mathcal{C}$. 
Accordingly, the edges of $G_t$ are also extended along the TTM dimension, yielding $E_{t0}^\mathrm{CC+},E_{t0}^\mathrm{CC-}, E_{t0}^\mathrm{CU}, E_{t0}^\mathrm{UC}$. 
Specifically, $E_{t0}^\mathrm{CC+}=\{((n,j),(n',j)): (n,n') \in E_{t}^\mathrm{CC+}, j\in[0:n_\mathrm{bas}]\}$, $E_{t0}^\mathrm{CU}=\{((n,j),n'): (n,n') \in E_{t}^\mathrm{CU}, j\in[0:n_\mathrm{bas}]\}$ with $E_{t0}^\mathrm{CC-}$ and $E_{t0}^\mathrm{UC}$ defined analogously. 
While preserving the original topology, $E_{t0}^\mathrm{CC+}\cup E_{t0}^\mathrm{CC-}$ connects the virtual nodes with the same virtual TTMs, and $E_{t0}^\mathrm{CU}\cup E_{t0}^\mathrm{UC}$ connects virtual nodes and their original futures nodes. 
Next, $f_k^\mathrm{L}$ produces the initial embedding $Z_{t0}=[\mathbf{z}_{t,0,cd}]_{(c,d)\in{\hat U_{t+2}}}$ where $\mathbf{z}_{t,0,cd} = \operatorname{Dropout}(\phi(W_{\mathrm{L},k}\mathbf{x}_{tcd}+ \mathbf b_{\mathrm{L},k})) \in \mathbb{R}^{n_\mathrm{hid}}$ 
for $(c,d)\in\hat U_{t+2}$. 
Here, $n_\mathrm{hid}\in\mathbb{N}$ is a hyperparameter.

\begin{figure}
    \centering
    \includegraphics[width=0.95\textwidth,keepaspectratio]{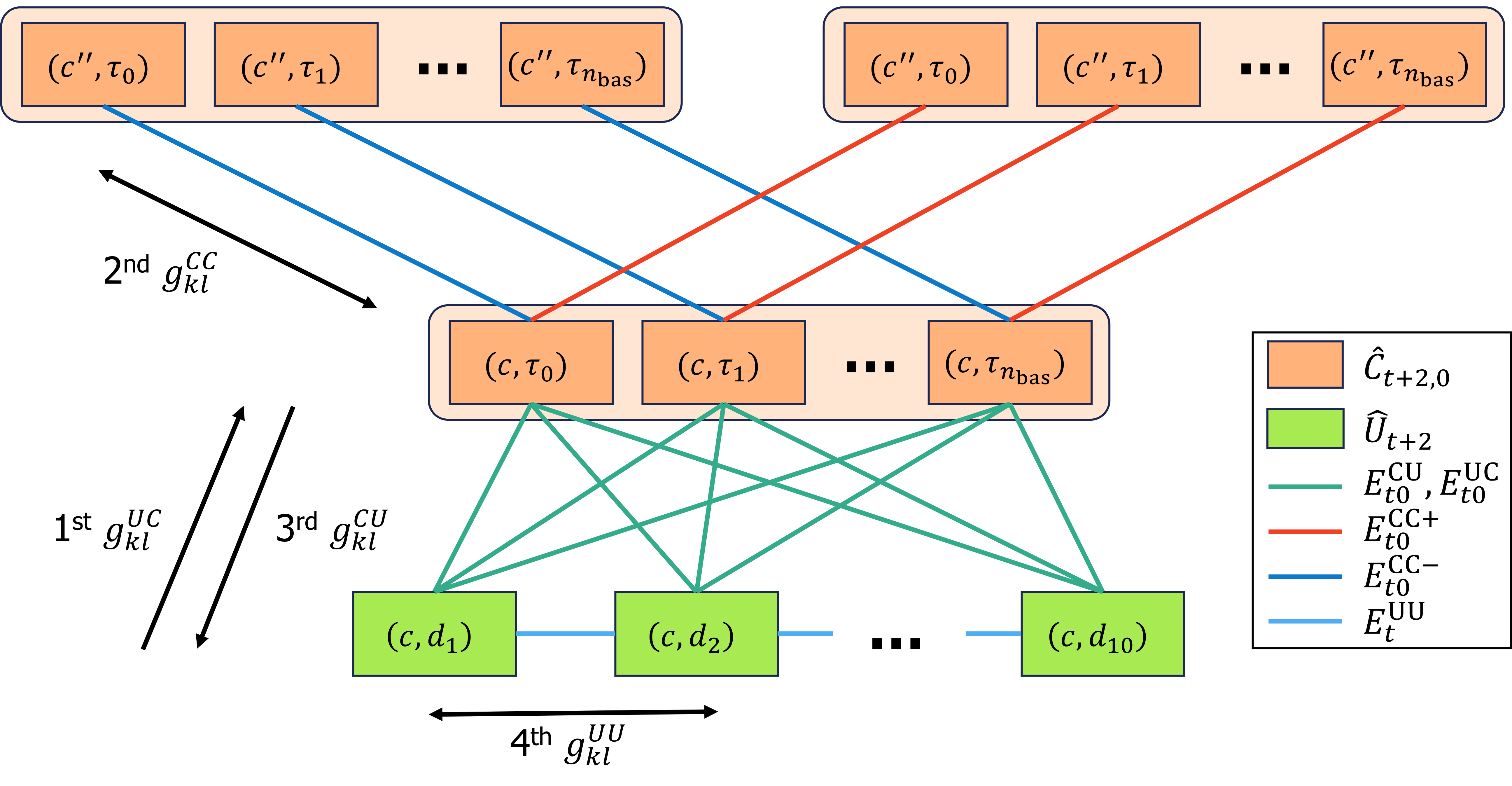}
    \captionof{figure}{Bi-Level Convolution}
    \label{fig:hgnn}
\end{figure}

\textbf{The bi-level convolution function} $f_{k}^\mathrm{B}(G_{t0})=(f_{k,l_\mathrm{con}}^\mathrm{B}\circ \dots \circ f_{k,1}^\mathrm{B}) (G_{t0})$ with $l_\mathrm{con}$ layers outputs $G_{t,l_\mathrm{con}}$, where each layer function $f_{kl}^\mathrm{B}$ maps $G_{t,l-1}$ to $G_{tl}$.  
Each layer output graph $G_{tl}$ shares the same node and edge sets but differs in its node embedding matrix $Z_{tl}=[\mathbf{z}_{tlcd}]_{(c,d)\in\hat U_{t+2}}\in \mathbb{R}^{ \hat{n}_{t+2} \times n_\mathrm{hid}}$, so that $G_{tl}=(\hat{C}_{t+2,0} \cup \hat{U}_{t+2},  E_{t0}^\mathrm{CC+} \cup E_{t0}^\mathrm{CC-} \cup E_{t0}^\mathrm{CU} \cup E_{t0}^\mathrm{UC}\cup E_{t}^\mathrm{UU}, Z_{tl})$. 
Each convolution layer function $f_{kl}^\mathrm{B}$ consists of four operations, illustrated in \Cref{fig:hgnn}: $f_{kl}^\mathrm{B} =g_{kl}^\mathrm{UU}\circ g_{kl}^\mathrm{CU} \circ g_{kl}^\mathrm{CC} \circ g_{kl}^\mathrm{UC}$ where all the operation outputs share the same node and edge sets but have distinct node embeddings. 
First, $g_{kl}^\mathrm{UC}$ elevates embeddings from lower-level contract nodes to TTM-aligned virtual futures nodes at the upper-level. 
Next, $g_{kl}^\mathrm{CC}$ propagates information among virtual futures contracts sharing the same TTM.  
Then, $g_{kl}^\mathrm{CU}$ transfers the aggregated upper-level information back to lower-level contract nodes. 
Finally, $g_{kl}^\mathrm{UU}$ propagates information across contracts with the nearest maturities within the same commodities. 
Throughout these steps, we distinguish contract relationships based on maturity differences and extract information accordingly via $g_{kl}^\mathrm{CC}$ and $g_{kl}^\mathrm{UU}$, in contrast to existing approaches \citep{hu2025graph,tan2024futures} that treat all relationships in the same manner regardless of maturity.
Each step is discussed next.

The elevating operator $g_{kl}^\mathrm{UC}$ transforms 
$Z_{t,l-1}=[\mathbf{z}_{t,l-1,cd}]_{(c,d)\in{\hat U_{t+2}}}\in \mathbb{R}^{\hat n_{t+2} \times n_\mathrm{hid}}$ 
into 
$Z_{tl}^\mathrm{UC}=[\mathbf{z}_{tlcj}^\mathrm{UC}]_{(c,j)\in\hat C_{t+2}\times [0:n_\mathrm{bas}]} \in \mathbb{R}^{\hat n_{t+2}(n_\mathrm{bas}+1) \times n_\mathrm{hid}}$, 
via linear interpolation 
\begin{equation}
\mathbf{z}_{tlcj}^\mathrm{UC}=\frac{\mathbf{z}_{t,l-1,cd^+}-\mathbf{z}_{t,l-1,cd^-}}{T_{cd^+}-T_{cd^-}}\cdot(t+\tau_j-T_{cd^-})+\mathbf{z}_{t,l-1,cd^-}  \in\mathbb{R}^{n_\mathrm{hid}} \label{eq:interpol1}
\end{equation}
for $t+\tau_j\in [T_{c,\min \hat D_{t+2,c}},T_{c,\max \hat D_{t+2,c}}]$, 
where $d^-,d^+\in \hat D_{t+2,c}$ minimize $|T_{cd^+}-T_{cd^-}|$ subject to $t+\tau_j\in [T_{cd^-},T_{cd^+}]$. 
Since futures contracts with similar TTMs share similar economic characteristics \citep{brennan1976supply,gibson1990stochastic,schwartz1997stochastic,casassus2005stochastic}, we use interpolation based on neighboring maturities.
For $t+\tau_j\notin [T_{c,\min \hat D_{t+2,c}},T_{c,\max \hat D_{t+2,c}}]$, we use the end point values for the extrapolation by setting $\mathbf{z}_{tlcj}^\mathrm{UC}=\mathbf{z}_{t,l-1,cd^{*}}$ with $d^{*}=\mathbb{I}_{\{\tau_j+t<T_{c,\min \hat D_{t+2,c}}\}} \cdot \min \hat D_{t+2,c}+\mathbb{I}_{\{\tau_j+t>T_{c,\max \hat D_{t+2,c}}\}} \cdot  \max \hat D_{t+2,c}$. 
Since both interpolation and extrapolation are linear, 
for each $c\in\hat C_{t+2}$, we can write
$[\mathbf{z}_{tlcj}^\mathrm{UC}]_{j\in [0:n_\mathrm{bas}]}
=\sum_{d\in \hat D_{t+2,c}} \mathbf{p}_{tcd} \mathbf{z}_{t,l-1,cd}^\top \in \mathbb{R}^{(n_\mathrm{bas}+1)\times n_\mathrm{hid}}$ 
for some constant vectors 
$\mathbf{p}_{tcd}\in\mathbb{R}^{n_\mathrm{bas}+1}$.
Note that $\mathbf{p}_{tcd}$ does not depend on $l$ and $Z_{t,l-1}$ but depends on $t, \tau_j,T_{cd^+}$, and $T_{cd^-}$.  
Thus, $\mathbf{p}_{tcd}$ can be precomputed before training and inference, and shares the same values for all layers. 
As a result of $g_{kl}^\mathrm{UC}$, each $\mathbf{z}_{tlcj}^\mathrm{UC}$ contains the information of a virtual future contract with TTM $\tau_j$ for commodity $c$.

The inter-commodity propagation $g_{kl}^\mathrm{CC}$ transforms $Z_{tl}^\mathrm{UC}$ into 
\begin{equation}
    Z_{tl}^{CC}=\phi([Z_{tl}^\mathrm{UC}~Z_{tl}^{CC+} ~Z_{tl}^{CC-}] W_{\mathrm{CC},kl} + \mathbf 1 \mathbf b_{\mathrm{CC},kl}^\top) \in\mathbb{R}^{\hat n_{t+2}(n_\mathrm{bas}+1)\times n_\mathrm{hid}}
\end{equation}
where $[Z_{tl}^\mathrm{UC}~Z_{tl}^{CC+} ~Z_{tl}^{CC-}]\in\mathbb{R}^{\hat n_{t+2}(n_\mathrm{bas}+1)\times (3n_\mathrm{hid})}$, 
$Z_{tl}^{CC+}=
\phi(\operatorname{CONV}(\hat{C}_{t+2,0},E^{CC+}_{t0},Z_{tl}^{UC};\theta_{CC,l}^+))
\in\mathbb{R}^{\hat n_{t+2}(n_\mathrm{bas}+1)\times n_\mathrm{hid}}$, 
and 
$Z_{tl}^{CC-}=
\phi(\operatorname{CONV}(\hat{C}_{t+2,0},E^{CC-}_{t0},Z_{tl}^{UC};\theta_{CC,l}^-))
\in\mathbb{R}^{\hat n_{t+2}(n_\mathrm{bas}+1)\times n_\mathrm{hid}}$. 
Here, $\mathrm{CONV}$ denotes a graph convolution, and  $\theta_{CC,l}^+$ and $\theta_{CC,l}^-$ are distinct learnable parameters. 
Since $E^{CC+}_{t0}$ and $E^{CC-}_{t0}$ connect nodes with the same virtual TTM, 
each node $(c,j^*) \in \hat{C}_{t+2,0}$ aggregates information $\mathbf{z}_{tlc'j^*}^\mathrm{UC}$ from other nodes with the same TTM $j^*$ but different commodities $c'$, while preserving the original topology of $E^{CC+}_{t}$ and $E^{CC-}_{t}$. 
In other words, $g_{kl}^\mathrm{CC}$ propagates TTM-aligned information $[\mathbf{z}_{tlcj}^\mathrm{UC}]_{j\in [0:n_\mathrm{bas}]}$ 
across commodities $c \in \hat C_{t+2}$.

The lowering function $g_{kl}^\mathrm{CU}$ converts $Z_{tl}^{CC}=[\mathbf{z}_{tlcj}^\mathrm{CC}]_{c\in\hat C_{t+2},j\in [0:n_\mathrm{bas}]} \in \mathbb{R}^{\hat n_{t+2}(n_\mathrm{bas}+1) \times n_\mathrm{hid}}$ into $Z_{tl}^{CU}=[\mathbf{z}_{tlcd}^\mathrm{CU}]_{(c,d)\in{\hat U_{t+2}}}\in \mathbb{R}^{\hat n_{t+2} \times n_\mathrm{hid}}$ via linear interpolation  
\begin{equation}
    \mathbf{z}_{tlcd}^\mathrm{CU}= \frac{\mathbf{z}_{tlc,j^{-}+1}^\mathrm{CC}-\mathbf{z}_{tlcj^{-}}^\mathrm{CC}}{\tau_{j^{-}+1}-\tau_{j^{-}}}\cdot(T_{cd}-t-\tau_{j^-})+\mathbf{z}_{tlcj^{-}}^\mathrm{CC} \in\mathbb{R}^{n_\mathrm{hid}} \label{eq:interpol2}
\end{equation}
where $j^-\in [0:n_\mathrm{bas})$ satisfies $T_{cd}-t \in [\tau_{j^-},\tau_{j^-+1}]$. 
Since the intervals $[\tau_{j^-},\tau_{j^-+1}]$ partition an interval $[0,\tau^{\max}]$, the minimization used for the indexes in \cref{eq:interpol1} is not required. 
Since this interpolation is linear, it can be precomputed, as discussed in the elevating operator \cref{eq:interpol1}.

The intra-commodity propagation $g_{kl}^\mathrm{UU}$ maps 
$Z_{tl}^{CU}$ to 
$Z_{tl}=[\mathbf{z}_{tlcd}]_{(c,d)\in{\hat U_{t+2}}}\in \mathbb{R}^{\hat n_{t+2} \times n_\mathrm{hid}}$ where
\begin{equation}
    \mathbf{z}_{tlcd}=\phi(W_{\mathrm{UU},kl}
    [\mathbf{z}_{tlcd}^\mathrm{CU};
    \mathbf{z}_{tlcd}^\mathrm{UU}]+\mathbf b_{\mathrm{UU},kl}) \in \mathbb{R}^{n_\mathrm{hid}} 
    \label{eq:intra-commodity-propagation}
\end{equation}
with
\begin{align}
    \mathbf{z}_{tlcd}^\mathrm{UU}
    &= \sum_{d'\in \{d'': (d'',d)\in E_{t}^\mathrm{UU}\} }  \mathbf{z}_{tlcdd'}^\mathrm{UU*} \in \mathbb{R}^{n_\mathrm{hid}}, \label{eq:intra-commodity-propagation-agg} 
    \\
    \mathbf{z}_{tlcdd'}^\mathrm{UU*}&=\phi(W_{\mathrm{UU*},kl} \cdot \operatorname{LayerNorm}( 
    \mathbf{z}_{tlcd}^\mathrm{CU}- \mathbf{z}_{tlcd'}^\mathrm{CU}
    )+\mathbf b_{\mathrm{UU*},kl}) \in \mathbb{R}^{n_\mathrm{hid}}, \label{eq:intra-commodity-propagation-diff} 
\end{align}
In \cref{eq:intra-commodity-propagation}, $\mathbf{z}_{tlcd}$ integrates the target contract embedding $\mathbf{z}_{tlcd}^\mathrm{CU}$ with neighboring information $\mathbf{z}_{tlcd}^\mathrm{UU}$. 
\Cref{eq:intra-commodity-propagation-agg} aggregates messages $\mathbf{z}_{tlcdd'}^{\mathrm{UU*}}$ from neighboring contracts in $E_t^{\mathrm{UU}}$. 
Each message $\mathbf{z}_{tlcdd'}^{\mathrm{UU*}}$ is computed from the normalized difference between the target and source embeddings as in \cref{eq:intra-commodity-propagation-diff} to capture relative deviations along the term structure within the same commodity. 
Consequently, $g_{kl}^\mathrm{UU}$ extracts intra-commodity information from nearest-maturity contracts of the same underlying commodity, reflecting the economic insight that futures with similar TTMs exhibit similar characteristics \citep{brennan1976supply,gibson1990stochastic,schwartz1997stochastic,casassus2005stochastic}.

Throughout the bi-level convolution $f_{k}^\mathrm{B}$, each futures contract node integrates its information with that of related nodes while explicitly accounting for their TTMs, in contrast to existing approaches \citep{hu2025graph,tan2024futures} that ignore the maturity dependence of these relationships. 
In their approaches, interactions between futures contracts are treated without differentiation by maturity, so information from contracts with similar maturities is processed in the same manner as that from contracts with dissimilar maturities. 
In our method, however, each convolution layer $f_{k,l}^\mathrm{B}$ uses two distinct message passing functions: $g_{kl}^{\mathrm{CC}}$, which propagates information across different commodities at the same TTM, and $g_{kl}^{\mathrm{UU}}$, which propagates information across nearby TTMs within the same commodity. 
By stacking multiple layers of $f_{k,l}^\mathrm{B}$, the model propagates information across contracts with larger maturity gaps through successive local interactions, thereby capturing higher-order dependencies over the joint commodity-maturity space. 
This design enables maturity-dependent propagation by explicitly distinguishing the relationships based on TTM while capturing inter-commodity and intra-commodity structures.

\textbf{The head} $f_{k}^\mathrm{H}(G_{tkl_\mathrm{con}})$ transforms the final node embedding matrix $Z_{t,l_\mathrm{con}}$ into a prediction vector $[\hat{Y}_{t+2,cd}]_{(c,d)\in \hat U_{t+2}}$. 
For each $(c,d)\in \hat U_{t+2}$, the head $f_{k}^\mathrm{H}(G_{tkl_\mathrm{con}})$ outputs $\hat{Y}_{t+2,cd} = W_{\mathrm{H},k}\mathbf{z}_{tl_\mathrm{con}cd}+\mathbf b_{\mathrm{H},k} \in \mathbb{R}$.  
During the training, we minimize the mean squared error between $\hat{Y}_{t+2,cd}$ and $Y_{t+2,cd}$ for $(c,d)\in \hat U_{t+2} \cap \tilde U_{t+2}$.

\section{Experiments}
\label{sec:experiments}

We conduct experiments to address the following questions. 
\begin{enumerate}[label=\textbf{(Q\arabic*)}, leftmargin=*]
    \item \label{rq:assumptions} Can CS strategies be more effective than LO strategies for statistical arbitrage in commodity futures markets in terms of risk and risk-adjusted return?   

    \item \label{rq:prediction} Are TTM-dependent interrelationships among futures contracts instrumental for statistical arbitrage? 

    \item \label{rq:architecture} Are both inter-commodity and intra-commodity relationships instrumental?
\end{enumerate}
Specifically, we verify \cref{eq:test1,eq:test2,eq:test3} in the data, together with trading examples of LO baselines, for \cref{rq:assumptions}; compare our method against benchmarks in prediction and trading for \cref{rq:prediction}; and conduct ablation studies for \cref{rq:architecture}.

\textbf{Baselines, Benchmarks, and Ablations}. 
We consider the following LO strategy baselines for \ref{rq:assumptions}. 
\begin{enumerate}[label=\textbf{(B\arabic*)}, leftmargin=*]
\item \label{ablation:bm1} $\mathbf{w}_t={\hat{\mathbf{V}}_t}/{\mathbf{1}^\top \hat{\mathbf{V}}_t}$ where $\hat{\mathbf{V}}_t =[\Phi(\hat Y_{tcd})]_{(c,d)\in\hat U_{t}}$
\item \label{ablation:bm2} $\mathbf{w}_t={\hat{\mathbf{V}}_t}/{\mathbf{1}^\top \hat{\mathbf{V}}_t}$ where
$\hat{\mathbf{V}}_t =\operatorname{cdf}(\hat Y_{tcd};[\hat Y_{tcd}]_{(c,d)\in\hat U_{t}})$
\item \label{ablation:bm3} $\mathbf{w}_t={\hat{\mathbf{V}}_t}/{\mathbf{1}^\top \hat{\mathbf{V}}_t}$ where 
$\hat{\mathbf{V}}_t =[\hat Y_{tcd} - \hat Y_{t,\min}]_{(c,d)\in\hat U_t}$ and $\hat Y_{t,\min}=\min \{\hat Y_{tc'd'}:(c',d')\in \hat U_t\}$
\item \label{ablation:bm4} $\mathbf{w}_t=\frac{1}{\hat n_{t}} \mathbf{1} + k\mathbf{w}^{\mathrm{CS}}_t$ where $k=\min\{\frac{1}{\hat n_t |w_{tcd}^{\mathrm{CS}}|}:(c,d)\in\hat U_t, w_{tcd}^{\mathrm{CS}}<0\}$ and $\mathbf{w}^{\mathrm{CS}}_t=[w_{tcd}^{\mathrm{CS}}]_{(c,d)\in\hat U_t}$ is a CS weight position vector for $\hat U_{t}$
\end{enumerate} 
The first equations of \cref{ablation:bm1,ablation:bm2,ablation:bm3,ablation:bm4} ensure the unit sum requirement of LO positions, by scaling for \cref{ablation:bm1,ablation:bm2,ablation:bm3} and by using the zero sum property of CS positions for \cref{ablation:bm4}. 
The where clauses guarantee the nonnegative weight requirement of LO positions.  
From \cref{eq:prob_int,eq:empirical_target}, the first baseline transforms $\hat Y_{tcd}$ into a prediction for $\widehat{\operatorname{cdf}}(\tilde R^\circ_{tcd};[\tilde R^\circ_{tcd}]_{(c,d)\in \hat{U}_{t} \cap \tilde U_t})$. 
The second baseline directly maps $\hat Y_{tcd}$ to its empirical cumulative distribution value within the cross section $[\hat Y_{tcd}]_{(c,d)\in\hat U_{t}}$. 
The third baseline subtracts the minimum value of the predictions. 
The last baseline is an inverse transformation of $\check{\mathbf{w}}$ in \Cref{prop:ir}. 
The four resulting baseline LO positions preserve the ordering of the prediction or weight values: $\Phi$ and $\operatorname{cdf}$ in the first two baselines are monotonically increasing, and the translations in the last two do not change the order. 
Thus, \cref{ablation:bm1,ablation:bm2,ablation:bm3,ablation:bm4} are natural transformations from a prediction vector to an LO position vector.

The benchmarks for \ref{rq:prediction} include ridge linear regression, LightGBM \citep{ke2017lightgbm}, and multi-layered perceptrons (MLP) to assess the importance of the information extracted from commodity futures interrelationships.    
These methods take the node features $\mathbf{x}_{tcd}$, whose components are defined in \cref{eq:normal-features}, as predictor vectors, and use $Y_{t+2,cd}$ defined in \cref{eq:prob_int} as the prediction target, consistent with our method. 
Ridge linear regression is adopted due to its strong theoretical grounding, robustness to multicollinearity, and role as a canonical linear benchmark.
As gradient-boosted tree methods have been shown to outperform neural networks on numerous tabular datasets \citep{ke2017lightgbm,grinsztajn2022tree,shwartz2022tabular,mcelfresh2023neural,shmuel2024comprehensive}, such as $\mathbf{x}_{tcd}$ in our study, we include LightGBM as a benchmark.  
We also include MLP as a fundamental neural network baseline. 
The MLP consists of $l_\mathrm{con}$ hidden layers with $n_\mathrm{hid}$ hidden units, followed by a final linear output layer.  
Since these benchmarks are non-graph methods, comparing them with our method shows the value of incorporating relational information in commodity futures.

We also include graph neural networks (GNNs) as benchmarks for \ref{rq:prediction} to assess whether the TTM-dependence is instrumental in extracting information from commodity futures relationships. 
To predict $Y_{t+2,cd}$ given $\mathcal{H}_t$, the GNNs use a graph $(\hat U_{t+2}, E_t^{UU+}\cup E_t^{UU-}, X_t)$, where nodes represent contracts, edges encode positive and negative correlations, and node features follow \cref{eq:normal-features}. 
Consistent with our method, we define
$E_t^{UU+}=
\{((c,d),(c',d'))\in \hat U_{t+2} \times \hat U_{t+2}:
\hat \rho_{tcdc'd'} \geq \rho^*
\}$ and 
$E_t^{UU-}=
\{((c,d),(c',d'))\in \hat U_{t+2} \times \hat U_{t+2}:
\hat \rho_{tcdc'd'} \leq -\rho^*
\}$. 
Here, $\hat \rho_{tcdc'd'}$ is the Pearson correlation computed from 
the accumulated pairs of $(\tilde{R}_{t''cd}^*,\tilde{R}_{t''c'd'}^*)$ over $t''\in[1:t]$, using only time points where both returns are available and at least $n_\mathrm{sam}^\mathrm{min}$ observations exist. 
Since GNNs capture contract relationships but ignore TTMs, they serve to evaluate the importance of TTM dependence.

The GNN consists of an initial embedding layer, $l_\mathrm{con}$ interim layers, and a final head. 
The initial layer embeds $X_t$ as $Z^G_{t,0}=\phi( X_t W_{\mathrm{G},k0}+ \mathbf 1 \mathbf b_{\mathrm{G},k0}^\top)\in\mathbb{R}^{\hat n_{t} \times n_\mathrm{hid}}$. 
Each interim layer propagates information over $E_t^{UU+}$ and $E_t^{UU-}$: 
for $l\in[1:l_\mathrm{con}]$, the $l$-th interim layer output is 
\begin{equation}
    Z_{t,l}^{G} = \phi \left(
    \left[
\phi(\operatorname{CONV}(\hat U_{t+2}, E_t^{UU+}, Z_{t,l-1}^{G}; \theta_{G,l}^-)),~ 
\phi(\operatorname{CONV}(\hat U_{t+2}, E_t^{UU-}, Z_{t,l-1}^{G}; \theta_{G,l}^+))
\right]
W_{G,kl}
+\mathbf 1  \mathbf b_{G,kl}^\top
    \right)\in\mathbb{R}^{\hat n_{t} \times n_\mathrm{hid}}
\end{equation}
where $\mathrm{CONV}$ denotes a graph convolution, and $\theta_{G,l}^-$ and $\theta_{G,l}^+$ are learnable parameters. 
A final head outputs a prediction vector $[\hat Y_{t+2,cd}]_{(c,d)\in \hat U_{t+2}}= Z_{t,l_\mathrm{con}}^{G} W_{G,k,H}+ \mathbf 1 \mathbf b_{G,k,H}^\top \in\mathbb{R}^{\hat n_{t}}$.

We consider the following ablations for \ref{rq:architecture}, keeping all other components unchanged: 
\begin{enumerate}[label=\textbf{(A\arabic*)}, leftmargin=*]
\item \label{ablation:a1} Replace $f_{kl}^\mathrm{B} =g_{kl}^\mathrm{UU}\circ g_{kl}^\mathrm{CU} \circ g_{kl}^\mathrm{CC} \circ g_{kl}^\mathrm{UC}$ with $f_{kl}^\mathrm{B} =g_{kl}^\mathrm{UU}$, 
\item \label{ablation:a2} Replace $f_{kl}^\mathrm{B} =g_{kl}^\mathrm{UU}\circ g_{kl}^\mathrm{CU} \circ g_{kl}^\mathrm{CC} \circ g_{kl}^\mathrm{UC}$ with $f_{kl}^\mathrm{B} = g_{kl}^\mathrm{CU} \circ g_{kl}^\mathrm{CC} \circ g_{kl}^\mathrm{UC}$.  
\end{enumerate}
These variants remove inter-commodity and intra-commodity information, respectively. 
Comparing them with the full model verifies the necessity of the two relationships in \ref{rq:architecture}.

\textbf{Experimental settings}. 
We use daily commodity futures price data from LSEG Datastream accessed via Wharton Research Data Services.\footnote{See Appendix \Cref{appendix:data_processing} for details on data processing. 
} 
The data span from Aug. 1977 to Dec. 2025 and consist of $\mathcal{C}$, $\mathcal{D}_c$ and $\mathcal{T}_c$ for each $c\in\mathcal{C}$, and $\mathcal{H}_t$ for futures contracts listed on CME, CBOT, NYMEX, COMEX, and eCBOT and traded in U.S. dollars.
Trading dates are defined as weekdays on which the number of traded contract tickers exceeds 50\% of the average over the preceding 365 days. 
Since future trading dates are unavailable, TTM and the difference of maturities are estimated using calendar dates. 
To ensure sufficient training samples, we set $t_1$ as the first trading date in 2016, define $t_k$ as the first trading date of each subsequent year, and set $n_{\mathrm{fit}}=10$.
We conduct stratified sampling for $\mathcal{T}_k^\mathrm{val}$ to mitigate the influence of period-specific market effects, treating each month as a separate stratum. 
We set $\tau^\mathrm{max}$ to one year, $n_\mathrm{bas}=52$ (the number of weeks in a year), and $n_\mathrm{sam}^\mathrm{min}=28$ that ensures 99\% of contracts to be traded in the testing period. 
Statistics of the trading universe are provided in Appendix \Cref{appendix:universe}. 
Since $\hat U_{t}$ can contain unobservable contracts, i.e., contracts with $\tilde V_{t-1,cd} \tilde V_{t,cd}=0$, $\tilde R_{t,cd}$ is unavailable for such contracts; thus, we compute \cref{eq:Gaam_def} as $\hat \Pi(\mathbf{w}_{tc})=\sum_{d\in \hat D_{tc} \cap \tilde D_{tc}} w_{tcd} \tilde R_{tcd}$.

\begin{table}
\caption{Hyperparameter Grids}
\label{tab:hyper_grids}
\begin{tabular}{l l}
\hline
Method & Hyperparameter Grid \\ \hline
Our method, GNN, \ref{ablation:a2} & $\mathrm{CONV}\in\{$GCN, SAGE, GAT$\}$, $\mathrm{\#params} \in   \{10^4,10^5\}$, $\rho^*\in\{0.1,0.2,0.3\}$, $l_\mathrm{conv}\in   \{1,2,3\}$ \\
Ridge & $\mathrm{regularizer} \in \left\{ 10^{-10 + 0.1\, i} : i \in   [0:200]\right\}$ \\
MLP,\ref{ablation:a1} & $\mathrm{\#params} \in \{10^4,10^5\}$, $l_\mathrm{conv}\in \{1,2,3\}$ \\
LightGBM &  \parbox[t]{10cm}{$\texttt{learning\_rate} \in \{0.02, 0.05, 0.1\}$, $\texttt{num\_leaves} \in \{127, 255\}$,   $\texttt{min\_child\_weight} \in \{100, 3000\}$, $\texttt{min\_child\_samples} \in \{20, 1000\}$, $\texttt{num\_round} \in   \{100, 500, 1000\}$, $(\texttt{top\_rate}, \texttt{other\_rate}) \in   \allowbreak \{(0.05,0.05), \allowbreak (0.05,0.10), \allowbreak (0.10,0.10),   \allowbreak (0.15,0.10), \allowbreak (0.15,0.25), \allowbreak (0.20,0.10),   \allowbreak (0.25,0.10)\}$} \\
\hline
\end{tabular}
\end{table}

We evaluate test performance using hyperparameters selected by validation for each period. 
For each $k\in[1:n_\mathrm{fit}]$ and each method, we train models over its hyperparameter grid on $\mathcal{T}_k^\mathrm{fit}$, select the best configuration based on validation performance on $\mathcal{T}_k^\mathrm{val}$, and evaluate on $\mathcal{T}_k^\mathrm{test}$. 
The grids are given in \Cref{tab:hyper_grids}. 
We use GCN \citep{kipf2017semi}, GAT \citep{velivckovic2018graph}, and SAGE \citep{hamilton2017inductive} due to their wide usage \citep{fey2019fast}. 
Although larger models may improve performance, the ranges of $\mathrm{\#params}, l_\mathrm{conv}, \rho^*$ are appropriate for the purpose of this study based on the selected configuration counts in Appendix \Cref{appendix:hyper_selection}. 
For each model, $n_\mathrm{hid}$ is chosen to achieve the closest match to the target $\mathrm{\#params}$. 
We set $\phi$ to the SiLU  \citep{ramachandran2017searching,hendrycks2023gaussian}, and use default $\mathrm{CONV}$ hyperparameters from the original papers. 
For LightGBM, we use the grid generated from the best configurations reported in its original paper appendix \citep{ke2017lightgbm}, set remaining parameters to defaults in \citep{lightgbm,lightgbm_pypi}, and apply early stopping as in \citep{ke2017lightgbm}.
In total, 8,790 models are trained across ten periods: 540 each for our method, GNN, and \ref{ablation:a2}; 2,010 for ridge; 60 each for MLP and \ref{ablation:a1}; and 5,040 for LightGBM.

% 10 3 2 3 3 = 540

% 10 201 = 2010

% 10 2 3 = 60

% 10 3 2 2 2 3 7 = 5040

\begin{figure}
    \centering
    \includegraphics[width=1.0\linewidth]{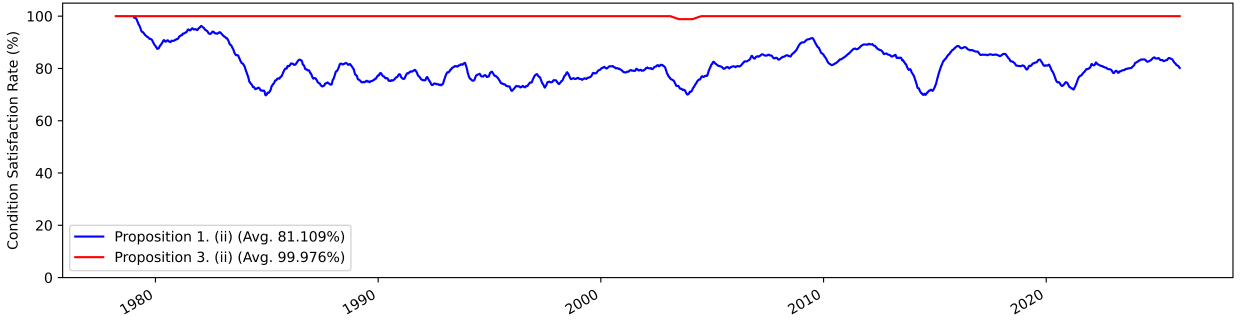}
    \caption{1-Year Rolling Average Rate of Satisfaction of Conditions (ii) in \Cref{prop:var,prop:low_delta}}
    \label{fig:prop-satisfaction}
\end{figure}

\begin{figure}
    \centering
    \includegraphics[width=1.0\linewidth]{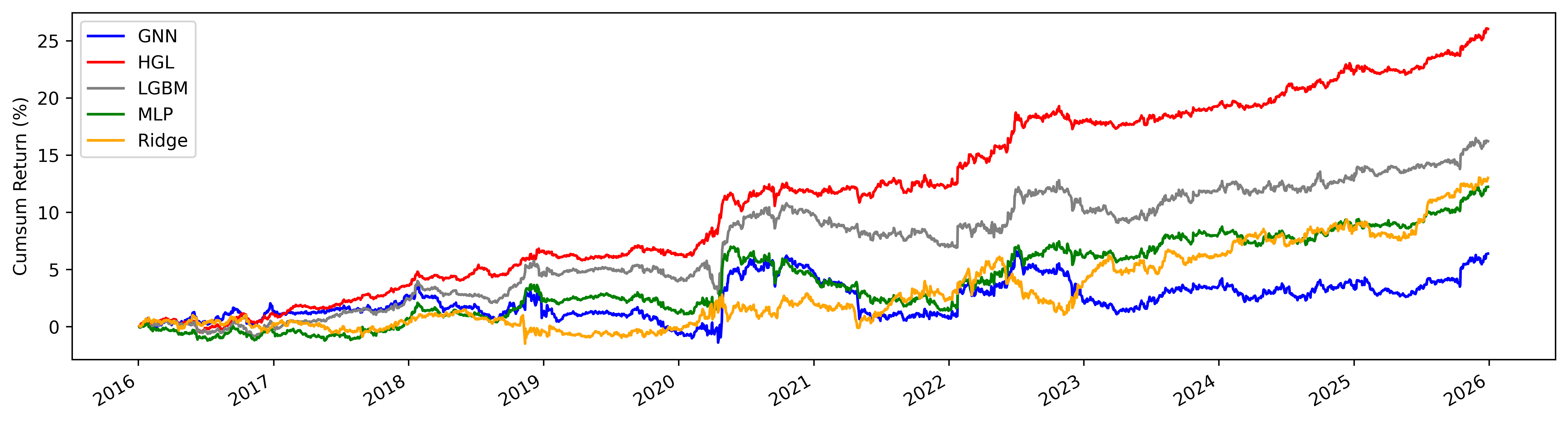}
    \caption{Cumulative Sum of Profit and Loss, under the trading scheme in which each day half of the portfolio weight is long and the other half is short. The jumps in 2020 and 2022 may be associated with COVID-19 and the Russo-Ukrainian war, respectively.}
    \label{fig:pnl}
\end{figure}

\begin{table}
\caption{
Trading Performance Across Methods 
(IR: information ratio; SR: Sortino ratio; Ret(\%): average daily return (\%); Vol(\%): daily return volatility (\%); Hit: hit ratio; MDD: maximum drawdown (absolute); Tvr: average daily turnover; Cor: correlation with the S\&P 500). 
Under typical margin ratios, the actual daily returns (except for the S\&P 500) can exceed the Ret values reported below by at least a factor of $8.33 = 1/0.12$. For example, the actual HGL daily return can exceed 0.08634\%. 
This is because the Ret values below are reported under the unit-margin assumption in \Cref{subsec:notation}, whereas typical margin ratios in practice are 3--12\% \citep{CMEFuturesMargins} and even lower for spreads \citep{CMEFuturesSpreads}.}
\label{tab:metrics}
\begin{tabular}{llrrrrrrrr}
\toprule
 &  & IR & SR & Ret (\%) & Vol (\%) & MDD & Hit & Tvr & Cor \\
Strategy & Method &  &  &  &  &  &  &  &  \\
\midrule
\multirow[t]{7}{*}{CS} & HGL & \textbf{0.08459} & \textbf{0.12409} & \textbf{0.01036} & 0.12249 & \textbf{0.02002} & \textbf{0.53959} & 0.55726 & -0.01850 \\
 & Ridge & 0.03299 & 0.04306 & 0.00517 & 0.15685 & 0.05074 & 0.50696 & 0.62866 & 0.04576 \\
 & MLP & 0.03233 & 0.04963 & 0.00487 & 0.15050 & 0.05609 & 0.51134 & 0.51328 & 0.02198 \\
 & LGBM & 0.04324 & 0.06834 & 0.00645 & 0.14922 & 0.03893 & 0.51612 & 0.65431 & 0.03262 \\
 & GNN & 0.01492 & 0.02131 & 0.00254 & 0.17031 & 0.05710 & 0.50378 & 0.59353 & 0.04165 \\
 & \ref{ablation:a1} & 0.07583 & 0.11027 & 0.00827 & \textbf{0.10904} & 0.02367 & 0.53442 & \textbf{0.53363} & \textbf{-0.03402} \\
 & \ref{ablation:a2} & 0.04420 & 0.06526 & 0.00609 & 0.13789 & 0.03028 & 0.52129 & 0.53653 & 0.01235 \\
\cline{1-10}
\multirow[t]{6}{*}{LO} & \ref{ablation:bm1} & 0.03266 & 0.04398 & 0.02605 & 0.79765 & 0.34652 & 0.53124 & 0.04710 & 0.30209 \\
 & \ref{ablation:bm2} & 0.03606 & 0.04927 & 0.02785 & \textbf{0.77239} & 0.31870 & 0.53004 & 0.27418 & 0.29924 \\
 & \ref{ablation:bm3} & 0.03825 & 0.05234 & 0.02994 & 0.78279 & \textbf{0.31691} & 0.52845 & 0.24609 & \textbf{0.29897} \\
 & \ref{ablation:bm4} & 0.03691 & 0.05018 & 0.02944 & 0.79753 & 0.33267 & 0.53084 & 0.21805 & 0.30171 \\
 & EW & 0.03224 & 0.04337 & 0.02576 & 0.79910 & 0.34957 & 0.53124 & \textbf{0.03947} & 0.30227 \\
 & S\&P 500 & \textbf{0.04836} & \textbf{0.05815} & \textbf{0.05525} & 1.14240 & 0.38249 & \textbf{0.54715} &  & 1.00000 \\
\bottomrule
\end{tabular}
\end{table}

\begin{table}[t]
\caption{MSE of Prediction Models. \textnormal{The first row indicates the model selected from GAT, GCN, and SAGE based on validation performance for each period.}}
\label{tab:mse}
\begin{tabular}{lrlr}
\toprule
Model & MSE & Model & MSE \\
\midrule
HGL & \textbf{1.125511} & GNN & 1.126562 \\
HGL-GAT & 1.125746 & GNN-GAT & 1.126513 \\
HGL-GCN & 1.125513 & GNN-GCN & 1.126587 \\
HGL-SAGE & 1.125342 & GNN-SAGE & 1.126562 \\
\midrule
LGBM & \textbf{1.126221} & Ridge & 1.127108 \\
MLP & 1.126497 &  &  \\
\bottomrule
\end{tabular}
\end{table}

\begin{table}
\caption{P-values from Paired Statistical Tests of MSE Differences between HGL and Benchmark Models. 
The first two tests assess the significance of the differences, while the last two columns evaluate the underlying assumptions: the Shapiro--Wilk test assesses the normality of the paired differences (for the t-test), and the symmetry test \citep{miao2006new} examines distributional symmetry (for the Wilcoxon test).}
\label{tab:pvals_model_tests}
\begin{tabular}{lrrrr}
\toprule
 & Paired T & Wilcoxon Signed Rank & Shapiro-Wilk & Symmetry \\
\midrule
GNN & 6.27e-05 & \textbf{8.14e-05} & 7.81e-10 & 8.74e-01 \\
GNN-GAT & \textbf{1.68e-04} & 2.95e-05 & 1.59e-02 & 9.99e-01 \\
GNN-GCN & 1.96e-05 & \textbf{4.88e-05} & 3.34e-05 & 9.99e-01 \\
GNN-SAGE & 6.27e-05 & \textbf{8.14e-05} & 7.81e-10 & 8.74e-01 \\
MLP & 6.22e-05 & \textbf{2.29e-05} & 6.80e-03 & 2.13e-01 \\
Ridge & \textbf{8.01e-09} & 5.52e-08 & 1.42e-01 & 2.43e-01 \\
LGBM & \textbf{2.14e-03} & 1.36e-03 & 1.32e-01 & 9.13e-01 \\
\bottomrule
\end{tabular}
\end{table}

\textbf{Experimental results} are summarized in \Cref{tab:mse,tab:pvals_model_tests,tab:metrics} and \Cref{fig:prop-satisfaction,fig:pnl}. 
HGL denotes our hierarchical graph learning method, with HGL-GAT/GCN/SAGE specifying $\mathrm{CONV}=\mathrm{GAT}/\mathrm{GCN}/\mathrm{SAGE}$; GNN-GAT/GCN/SAGE are defined analogously. 
HGL and GNN are the best convolution methods selected from GAT, GCN, and SAGE per period based on validation performance. 
In \Cref{tab:metrics} and \Cref{fig:pnl}, the daily CS position vectors are computed via \cref{eq:group-neutralized-position-vector1,eq:group-neutralized-position-vector2,eq:group-neutralized-position-vector3} where $\hat Y_{tcd}$ are the corresponding model predictions. 
In \Cref{tab:metrics}, the daily LO position vectors (except for the equal-weight (EW) strategy and S\&P 500) are computed from the predictions or CS weights of HGL. 
The EW strategy assigns equal positive weights across contracts. 
In \Cref{tab:metrics}, IR and SR measure risk-adjusted return; Vol and MDD measure risk; Tvr proxies transaction costs; and Cor measures diversification benefits relative to the S\&P 500 (lower is better).\footnote{Additionally, yearly metrics are provided in Appendix \Cref{appendix:annual}.}  
HGL achieves the highest risk-adjusted returns, measured by IR and SR, indicating the strongest statistical arbitrage performance.\footnote{Investment objectives in finance typically focus on maximizing risk-adjusted returns \citep{sharpe1966mutual,sharpe1998sharpe}.} 
We address each of \cref{rq:assumptions,rq:prediction,rq:architecture} based on these results.

\ref{rq:assumptions}
We find strong empirical support for the effectiveness of CS strategies over LO strategies for statistical arbitrage, as the assumptions in \Cref{prop:var,prop:low_delta} hold with high probability.  
Specifically, on average, 81.109\% of $(t,c)$ pairs in our trading universe satisfy both \cref{eq:test1,eq:test2} associated with \Cref{prop:var}, and 99.976\% satisfy \cref{eq:test3} for \Cref{prop:low_delta}.  
In \Cref{fig:prop-satisfaction}, the conditions in \Cref{prop:var,prop:low_delta} hold for at least 70\% and 99\%, respectively, indicating that CS strategies exhibit lower variance and delta in a large fraction of cases. 
% lower variance and delta for CS strategies in a large fraction of cases.  
Moreover, by \Cref{prop:ir} (ii), if an LO strategy achieves sufficiently high expected returns as characterized therein, then the demeaned CS strategy can outperform the LO strategy, as 
$\operatorname{Var}(\Pi(\mathbf{w}))
>\operatorname{Var}(\Pi(\check{\mathbf{w}}))$
holds in a large fraction of cases. 
Note that these percentages do not imply CS underperformance otherwise; rather, they indicate that theoretical superiority is not established in those cases.

Furthermore, \Cref{tab:metrics} provides an empirical example of the effectiveness of CS trading. 
HGL achieves more than twice the IR and SR, more than six times lower volatility, and more than  fifteen times lower MDD than the LO baselines \cref{ablation:bm1,ablation:bm2,ablation:bm3,ablation:bm4} and the EW strategy. 
These results are consistent with the propositions, showing that CS strategies can deliver both lower risk and higher risk-adjusted return than LO strategies. 
Moreover, HGL outperforms the S\&P 500, with over 75\% higher IR, 113\% higher SR, 
nine times lower volatility, and nineteen times lower MDD, indicating strong investment potential. 
It also exhibits near-zero correlation with the S\&P 500, implying substantial diversification benefits  when invested alongside traditional equities. 
At typical margin ratios employed by CME Group \citep{CMEFuturesMargins,CMEFuturesSpreads}, HGL's daily return is 188\% and 56\% higher than those of the LO baselines and the S\&P 500, respectively.

\ref{rq:prediction}
From the outperformance of HGL in \Cref{fig:pnl,tab:metrics,tab:mse,tab:pvals_model_tests}, we find that TTM-dependent interrelationships between futures contracts are instrumental for both trading and prediction. 
The cumulative sum P\&L of HGL is stable and upward-sloping in \Cref{fig:pnl}, compared to the other benchmarks. 
In \Cref{tab:metrics}, HGL consistently outperforms the non-graph CS benchmarks (Ridge, MLP, LGBM, and GNN) across all the key metrics: higher IR, SR, daily return, and hit ratio; and lower volatility and MDD. 
These results stem from the better prediction of HGL. 
All HGL models achieve lower MSEs (\Cref{tab:mse}) than the benchmarks, with differences statistically significant at the 0.01 level (\Cref{tab:pvals_model_tests}) based on paired tests listed in \citealp{hong2025graph} using daily MSEs. 
% \red{Additional tests for HGNNs with specific convolutions confirm consistent outperformance (Appendix \Cref{appendix:stat_tests}). }
Since HGL outperforms the non-graph methods, futures interrelationships are informative; since it also outperforms GNNs, TTM dependence is important for capturing these relationships.

\ref{rq:architecture}
We confirm that both inter- and intra-commodity relationships are instrumental. 
In \Cref{tab:metrics}, although \ref{ablation:a1} performs better than \ref{ablation:a2}, neither achieves IR and SR comparable to HGL, indicating that both types of relationships are informative. 
Hence, both are important for extracting information from futures contract relationships.

\textbf{Future work.}
Although our method generates statistical arbitrage and demonstrates the effectiveness of CS strategies and the importance of TTM-dependent futures interrelationships, there remains room for future research. 
Due to limited and contract-specific data availability, this study does not account for margin ratios or transaction costs; incorporating these factors is an important direction. 
While HGL captures the TTM-dependent interrelationship information, the exact nature of the extracted information remains unclear due to the black-box nature of deep learning methods, motivating further study. 
Although CS strategies can outperform LO strategies, their combination may yield additional gains under certain conditions, which presents an interesting direction for further study.

\section{Conclusion}
This paper proposes a hierarchical graph learning approach for CS strategies in commodity futures markets, bridging two key gaps in the machine-learning literature: (i) the absence of learning-based methods for CS strategies in futures markets, and (ii) the lack of consideration for maturity-dependent interrelationships among commodity futures. 
We establish the efficacy of CS strategies by proving that they can exhibit higher risk-adjusted returns and lower risk than long-only strategies.
We also present a method to transform predictions into CS positions. 
Next, we develop a hierarchical graph learning method that predicts futures price movements by leveraging TTM-dependent interrelationships among futures, thereby yielding a trading algorithm.
The experiments demonstrate that the hierarchical graph learning and resulting CS positions outperform benchmarks in prediction and trading, respectively. 
We find that maturity-dependent interrelationships across commodity futures are important and that CS trading based on  hierarchical graph learning is effective for statistical arbitrage.

\newpage

\section*{Declaration on the Use of Generative AI and AI-Assisted Technologies in Manuscript Preparation}
In preparing this work, the authors used ChatGPT to assist with English-language editing, such as grammar correction and translation.

\bibliographystyle{cas-model2-names}

\bibliography{00.sections/98.references}

\appendix

\newpage
\renewcommand{\theequation}{\Alph{section}.\arabic{equation}}
\renewcommand{\thefigure}{\Alph{section}.\arabic{figure}}
\renewcommand{\thetable}{\Alph{section}.\arabic{table}}

\setcounter{equation}{0}
\setcounter{figure}{0}
\setcounter{table}{0}
\section{Proofs of Propositions}
\label{appendix:proofs}

\subsection{Proof of Proposition~\ref{prop:var}}
\label{proof:var2}
\begin{proof}
Fix $t\in\mathbb{N}$ and $c \in \hat C_t$. 
We suppress the subscripts $(t,c)$ and write $w_{d}=w_{tcd}$, $\sigma_{dd'}=\sigma_{tcdd'}$, $\rho_{dd'}=\rho_{tcdd'}$, $\hat n=\hat n_{tc}$,  $\kappa=\kappa_{tc}$, $R_d = R_{tcd}$, $\hat D= \hat D_{tc}$,  $\mathcal W_{\mathrm{LO}} = \mathcal W_{tc,\mathrm{LO}}$, and $\mathcal W_{\mathrm{CS}} = \mathcal W_{tc,\mathrm{CS}}$. 
Let $\sigma_d=\sqrt{\sigma_{dd}}$, 
$\sigma_{\min}=\min_{d\in \hat D} \sigma_d$, 
$\sigma_{\max}=\max_{d\in \hat D} \sigma_d$, 
and $\rho^*=\min_{d,d'\in \hat D}\rho_{dd'}$. 
Let $\mathbf{w}=[w_{d}]_{d\in \hat D}\in \mathcal W_{\mathrm{LO}}$  and
$\mathbf{w}'=[w'_{d}]_{d\in \hat D}\in \mathcal W_{\mathrm{CS}}$ be arbitrary.

(i) Assume that $\rho^* > ({\kappa^2-2\norm{\mathbf{w}}_2^2})({3-2\norm{\mathbf{w}}_2^2})^{-1}$ and $\rho^* \geq 0$.  
We obtain \cref{eq:prop-var-1-lo1} from \cref{eq:Gaam_def} and \cref{eq:prop-var-1-lo2} from the definition of $\rho^*$. 
Pulling out $\rho^*\sum_{d\in \hat D} w_{d}^2 \sigma_d^2$ from the first term of \cref{eq:prop-var-1-lo2} and adding it to the second term of \cref{eq:prop-var-1-lo2}, we have \cref{eq:prop-var-1-lo3} since $\rho^*\sum_{d,d'\in \hat D} w_d w_{d'} \sigma_d \sigma_{d'}=\rho^*(\sum_{d\in \hat D} w_{d} \sigma_d)^2$. 
Since $\sigma_d^2 \geq \sigma_{\min}^2$ for all $d\in \hat D$, and $\mathbf{1}^\top \mathbf{w}=1$, we obtain \cref{eq:prop-var-1-lo4}. 
\begin{align}
\operatorname{Var}(\Pi(\mathbf{w}))&=
\sum_{d,d' \in  \hat D} w_{d} w_{d'} \sigma_{dd'}=\sum_{d\in \hat D} w_{d}^2 \sigma_d^2+\sum_{d,d' \in  \hat D:d\neq d'} w_d w_{d'} \sigma_d \sigma_{d'} \rho_{dd'} \label{eq:prop-var-1-lo1}
\\ &\geq \sum_{d\in \hat D} w_{d}^2 \sigma_d^2 + \rho^*\sum_{d,d'\in \hat D:d\neq d'} w_d w_{d'} \sigma_d \sigma_{d'} \label{eq:prop-var-1-lo2}   
\\ &= (1-\rho^*)\sum_{d\in \hat D} w_{d}^2 \sigma_d^2 + \rho^*(\sum_{d\in \hat D} w_{d} \sigma_d)^2 \label{eq:prop-var-1-lo3}   
\\ & \geq \sigma_{\min}^2 \left((1-\rho^*) \norm{\mathbf{w}}^2_2 +\rho^* \right)
\label{eq:prop-var-1-lo4}   
\end{align}

We obtain \cref{eq:prop-var-1-cs1} from \cref{eq:Gaam_def}. 
We have \cref{eq:prop-var-1-cs2} since $w'_{d}=[w'_{d}]^+ - [-w'_{d}]^+$. 
Since 
$\sigma_d \sigma_{d'} \rho^* \leq \sigma_d \sigma_{d'} \rho_{dd'} = \sigma_{dd'} \leq \sigma_d \sigma_{d'}$, we obtain \cref{eq:prop-var-1-cs4}. 
Since $\mathbf{w}'$ is a CS position weight vector, we have $\sum_{d\in \hat D} [w'_{d}]^+=\sum_{d\in \hat D} [-w'_{d}]^+=1/2$ from \cref{eq:symmetric_weight}. 
Then, we have $\sigma_{\max}^2/4 = \sigma_{\max}^2 ( \sum_{d \in  \hat D}  [w'_{d}]^+ )^2 \geq ( \sum_{d \in  \hat D} \sigma_d [w'_{d}]^+ )^2$ since $\sigma_d\leq \sigma_{\max}$, which yields the upper bound of the first term of \cref{eq:prop-var-1-cs5}. 
Analogously, we obtain $\sigma_{\max}^2/4 \geq ( \sum_{d \in  \hat D} \sigma_d [-w'_{d}]^+ )^2$, which yields the upper bound of the second term of \cref{eq:prop-var-1-cs5}.  
Therefore, we obtain \cref{eq:prop-var-1-cs6} since $\sum_{d \in  \hat D} \sigma_d [w'_{d}]^+ \geq \sigma_{\min} \sum_{d \in  \hat D}  [w'_{d}]^+=\sigma_{\min}/2$ and $\sum_{d \in  \hat D} \sigma_d [-w'_{d}]^+ \geq \sigma_{\min}/2$.

\begin{align}
\operatorname{Var}(\Pi(\mathbf{w}'))
&= \sum_{d,d' \in  \hat D} w'_{d} w'_{d'} \sigma_{dd'} \label{eq:prop-var-1-cs1}
\\&= \sum_{d,d' \in  \hat D} \sigma_{dd'}
\big([w'_{d}]^+ - [-w'_{d}]^+\big)
\big([w'_{d'}]^+ - [-w'_{d'}]^+\big) \label{eq:prop-var-1-cs2}
\\&= \sum_{d,d' \in  \hat D} \sigma_{dd'} [w'_{d}]^+ [w'_{d'}]^+
+ \sum_{d,d' \in  \hat D} \sigma_{dd'} [-w'_{d}]^+ [-w'_{d'}]^+ 
- 2 \sum_{d,d' \in  \hat D} \sigma_{dd'} [w'_{d}]^+ [-w'_{d'}]^+ \label{eq:prop-var-1-cs3}
\\&\le
\sum_{d,d' \in  \hat D} \sigma_d \sigma_{d'} [w'_{d}]^+ [w'_{d'}]^+
+ \sum_{d,d' \in  \hat D} \sigma_d \sigma_{d'} [-w'_{d}]^+ [-w'_{d'}]^+ 
- 2 \rho^*
\sum_{d,d' \in  \hat D} \sigma_d \sigma_{d'} [w'_{d}]^+ [-w'_{d'}]^+ 
\label{eq:prop-var-1-cs4}
\\&=
\left( \sum_{d \in  \hat D} \sigma_d [w'_{d}]^+ \right)^2
+
\left( \sum_{d \in  \hat D} \sigma_d [-w'_{d}]^+ \right)^2 
- 2 \rho^*
\left( \sum_{d \in  \hat D} \sigma_d [w'_{d}]^+ \right)
\left( \sum_{d \in  \hat D} \sigma_d [-w'_{d}]^+ \right) \label{eq:prop-var-1-cs5}
\\&\leq
\sigma_{\max}^2 (\frac{1}{4}+\frac{1}{4})-2\rho^* \cdot \frac{\sigma_{\min}}{2} \cdot \frac{\sigma_{\min}}{2}= \frac{1}{2} (\sigma_{\max}^2-\rho^*\sigma_{\min}^2) \label{eq:prop-var-1-cs6}
\end{align}

Since $\mathbf{1}^\top \mathbf{w}=1$ and $w_d \geq 0$, we have 
\begin{equation}
\norm{\mathbf{w}}_2^2 \leq (\mathbf{1}^\top \mathbf{w})^2 = 1 \label{eq:prop-var-1-w-bd}    
\end{equation}
and thus $3-2\norm{\mathbf{w}}_2^2>0$. 
From this, since $\rho^* > ({\kappa^2-2\norm{\mathbf{w}}_2^2})({3-2\norm{\mathbf{w}}_2^2})^{-1}$, we obtain 
\begin{align}
\rho^*(3-2\norm{\mathbf{w}}_2^2)&>\kappa^2-2\norm{\mathbf{w}}_2^2, \label{eq:prop-var-1-bd1}
\\ 
2 \left(\norm{\mathbf{w}}_2^2+\rho^*(1-\norm{\mathbf{w}}_2^2)\right) &> \kappa^2-\rho^*, \label{eq:prop-var-1-bd2}
\\
\sigma_{\min}^2 \left((1-\rho^*) \norm{\mathbf{w}}^2_2 +\rho^* \right) &> \frac{1}{2} (\sigma_{\max}^2-\rho^*\sigma_{\min}^2), \label{eq:prop-var-1-bd3}
\end{align}
where the last inequality follows from the definition of $\kappa^2$. 
From \cref{eq:prop-var-1-lo4,eq:prop-var-1-cs6,eq:prop-var-1-bd3}, we obtain $\operatorname{Var}(\Pi(\mathbf{w})) > \operatorname{Var}(\Pi(\mathbf{w}'))$.

(ii) Assume that $\min_{d,d'\in \hat D}\rho_{tcdd'} \geq 0$ and $\min_{d,d'\in \hat D}\rho_{tcdd'} > (\kappa^2 \hat n - 2)(3 \hat n-2)^{-1}$. 
We have $\kappa^2 < 3$ since $(\kappa^2 \hat n - 2)(3 \hat n-2)^{-1}=1+\frac{(\kappa^2-3)\hat n}{3 \hat n-2}$, $\hat n$ is natural, and every correlation coefficient does not exceed one. 
Since $\mathbf 1^\top \mathbf w =1$, 
by the Cauchy-Schwarz inequality and \cref{eq:prop-var-1-w-bd}, 
we obtain $1/\hat n \leq \norm{\mathbf{w}}_2^2 \leq 1$. 
We thus have
\begin{align}
    3-2\hat n^{-1} &\geq  3-2\norm{\mathbf{w}}_2^2, \\
    (\kappa^2 \hat n - 2)(3 \hat n -2)^{-1}=1+\frac{\kappa^2-3}{3-2 \hat n^{-1}}  &\geq \frac{\kappa^2-3}{3-2\norm{\mathbf{w}}_2^2}+1
    =({\kappa^2-2\norm{\mathbf{w}}_2^2})({3-2\norm{\mathbf{w}}_2^2})^{-1}. 
\end{align}
From this and the assumption, we obtain $\min_{d,d'\in \hat D}\rho_{tcdd'} > (\kappa^2 \hat n - 2)(3 \hat n-2)^{-1} \geq ({\kappa^2-2\norm{\mathbf{w}}_2^2})({3-2\norm{\mathbf{w}}_2^2})^{-1}$. 
From (i), we obtain $\operatorname{Var}(\Pi(\mathbf{w})) > \operatorname{Var}(\Pi(\mathbf{w}'))$. 
\end{proof}

\subsection{Proof of Proposition~\ref{prop:ir}}
\label{proof:ir2}
\begin{proof} 
Let $\mathbf{w}\in\mathcal{W}_{tc,\mathrm{LO}}$, $\bar{\mathbf{w}}=\mathbf 1/\hat n_{tc} \in\mathcal{W}_{tc,\mathrm{LO}}$, and $\check{\mathbf{w}}=(\mathbf{w}-\bar{\mathbf{w}})/\norm{\mathbf{w}-\bar{\mathbf{w}}}_1\in\mathcal{W}_{tc,\mathrm{CS}}$ with $\mathbf{w}\neq \bar{\mathbf{w}}$. 
Since both expectation and \cref{eq:Gaam_def} are linear, we have
\begin{align}
&\mathbb{E}\Pi(\mathbf{w}-\bar{\mathbf{w}})=\mathbb{E}\Pi(\mathbf{w})-\mathbb{E}\Pi(\bar{\mathbf{w}}), \label{eq:prop-e-w-check0}\\
&\mathbb E \Pi(\check{\mathbf w})
=\frac{\mathbb{E}\Pi(\mathbf{w}-\bar{\mathbf{w}})}{\norm{\mathbf{w}-\bar{\mathbf{w}}}_1}
=\frac{\mathbb{E}\Pi(\mathbf{w})-\mathbb{E}\Pi(\bar{\mathbf{w}})}{\norm{\mathbf{w}-\bar{\mathbf{w}}}_1}. \label{eq:prop-e-w-check}
\end{align}
Due to the linearity of \cref{eq:Gaam_def} and the property of variance, we have
\begin{equation}
\operatorname{Var}(\Pi(\check{\mathbf{w}}))
=\operatorname{Var}\left(
    \frac{\Pi (\mathbf{w}-\bar{\mathbf{w}})}{\norm{\mathbf{w}-\bar{\mathbf{w}}}_1}
    \right)
=\frac{\operatorname{Var}(\Pi(\mathbf{w}-\bar{\mathbf{w}}))}{\norm{\mathbf{w}-\bar{\mathbf{w}}}_1^2}.    
\label{eq:prop-var-w-check}
\end{equation} 
From \cref{eq:prop-e-w-check,eq:prop-var-w-check}, since $\operatorname{Var}(\Pi(\mathbf{w}-\bar{\mathbf{w}})) > 0$, we have
\begin{equation}
\operatorname{IR}(\Pi(\check{\mathbf{w}}))
=\frac{\mathbb{E}\Pi(\check{\mathbf{w}})}{\sqrt{\operatorname{Var}(\Pi(\check{\mathbf{w}}))}}
=\frac{\mathbb{E}[\Pi(\mathbf{w}-\bar{\mathbf{w}})]/\norm{\mathbf{w}-\bar{\mathbf{w}}}_1}{
\sqrt{\operatorname{Var}(\Pi(\mathbf{w}-\bar{\mathbf{w}}))/\norm{\mathbf{w}-\bar{\mathbf{w}}}_1^2}}
=\frac{\mathbb{E}\Pi(\mathbf{w}-\bar{\mathbf{w}})}{\sqrt{\operatorname{Var}(\Pi(\mathbf{w}-\bar{\mathbf{w}}))}}
=\operatorname{IR}(\Pi(\mathbf{w}-\bar{\mathbf{w}})). \label{eq:prop-ir-w-check}  
\end{equation}

(i) Assume that \cref{eq:prop-ir-cond1} holds, $\operatorname{Var}(\Pi(\mathbf{w}-\bar{\mathbf{w}})) > 0$, and $\operatorname{Var}(\Pi(\mathbf{w})) > 0$.  
From \cref{eq:prop-ir-cond1,eq:prop-e-w-check0}, we obtain
$\operatorname{IR}(\Pi(\mathbf{w}-\bar{\mathbf{w}}))
=
\frac{\mathbb{E}\Pi(\mathbf{w}-\bar{\mathbf{w}})}{\sqrt{\operatorname{Var}(\Pi(\mathbf{w}-\bar{\mathbf{w}}))}}
>
\frac{\red{[\mathbb{E}\Pi(\mathbf{w})]^+}}{\sqrt{\operatorname{Var}(\Pi(\mathbf{w})}}
=\red{[\operatorname{IR}(\Pi(\mathbf w))]^+}$, 
yielding $\operatorname{IR}(\Pi(\check{\mathbf{w}}))>\red{[\operatorname{IR}(\Pi(\mathbf{w}))]^+}$ due to \cref{eq:prop-ir-w-check}.

(ii) 
Assume that \cref{eq:prop-ir-cond2} holds, $\operatorname{Var}(\Pi(\mathbf{w}-\bar{\mathbf{w}})) > 0$, and $\operatorname{Var}(\Pi(\mathbf{w}))
>\operatorname{Var}(\Pi(\check{\mathbf{w}}))$.   
From \cref{eq:prop-ir-w-check}, we have the first two equalities of \cref{eq:proof-var-2-final}. 
By the assumption and \cref{eq:prop-var-w-check}, 
we obtain 
$\operatorname{Var}(\Pi(\mathbf{w}))
>\operatorname{Var}(\Pi(\check{\mathbf{w}}))
=
\frac{\operatorname{Var}(\Pi(\mathbf{w}-\bar{\mathbf{w}}))}{\norm{\mathbf{w}-\bar{\mathbf{w}}}_1^2}, $
which yields the inequality between the third and fourth terms in \cref{eq:proof-var-2-final} since $\mathbb{E}[\Pi(\mathbf{w}-\bar{\mathbf{w}})]>0$ by \cref{eq:prop-ir-cond2,eq:prop-e-w-check0}. 
By \cref{eq:prop-ir-cond2,eq:prop-e-w-check0}, we have 
$\mathbb{E} \Pi(\mathbf{w}-\bar{\mathbf{w}})
>
\norm{\mathbf{w}-\bar{\mathbf{w}}}_1 \red{[\mathbb{E}\Pi(\mathbf{w})]^+}$, yielding 
the inequality between the fourth and fifth terms in \cref{eq:proof-var-2-final}. 
From the following, we have $\operatorname{IR}(\Pi(\check{\mathbf{w}}))>\red{[\operatorname{IR}(\Pi(\mathbf{w}))]^+}$: 
\begin{equation}
 \operatorname{IR}(\Pi(\check{\mathbf{w}}))
=\operatorname{IR}(\Pi(\mathbf{w}-\bar{\mathbf{w}}))
=\frac{\mathbb{E}[\Pi(\mathbf{w}-\bar{\mathbf{w}})]}{\sqrt{\operatorname{Var}(\Pi(\mathbf{w}-\bar{\mathbf{w}}))}}
>\frac{\mathbb{E}[\Pi(\mathbf{w}-\bar{\mathbf{w}})]}{\norm{\mathbf{w}-\bar{\mathbf{w}}}_1 \sqrt{\operatorname{Var}(\Pi(\mathbf{w}))}}
>
\frac{\norm{\mathbf{w}-\bar{\mathbf{w}}}_1 \red{[\mathbb{E}\Pi(\mathbf{w})]^+}}{\norm{\mathbf{w}-\bar{\mathbf{w}}}_1 \sqrt{\operatorname{Var}(\Pi(\mathbf{w}))}}
=\red{[\operatorname{IR}(\Pi(\mathbf{w}))]^+}.   
\label{eq:proof-var-2-final}
\end{equation}

(iii)
We assume that \cref{eq:prop-ir-cond0} holds, $\operatorname{Var}(\Pi(\mathbf{w}-\bar{\mathbf{w}})) = 0$, and $\operatorname{Var}(\Pi(\mathbf{w})) > 0$. 
We have $\mathbb E \Pi (\check{\mathbf w})>0$ by \cref{eq:prop-ir-cond0,eq:prop-e-w-check} and $\operatorname{Var}(\Pi(\check{\mathbf{w}}))=0$ by \cref{eq:prop-var-w-check} and the zero-variance assumption. 
Let $\mathbf{w}=[w_{cd}]_{d \in \hat D_{tc}}$. 
From \cref{eq:Gaam_def}, 
$\mathbb E\Pi(\mathbf w) 
= \mathbb E \sum_{d\in \hat D_{tc}} w_{cd }R_{tcd}
< \sum_{d\in \hat D_{tc}} \mathbb E R_{tcd} < \infty$,
where the first inequality is because $\|\mathbf{w}\|_1=1$ and $w_{cd}\ge0$,
and the last inequality follows from \cref{eq:finite_return}. 
Thus, we obtain $\operatorname{IR}(\Pi(\mathbf w))<\infty$ since $\operatorname{Var}(\Pi(\mathbf{w})) > 0$.

\end{proof}

\subsection{Proof of Proposition~\ref{prop:low_delta}}
\label{proof:low_delta}
\begin{proof}
Fix $t\in\mathbb{N}$ and $c\in \hat C_t$. 

\textnormal{(i)} 
Let $\mathbf{w}\in\mathcal{W}_{tc,\mathrm{LO}}$, 
$\bar{\mathbf{w}}=\mathbf 1/\hat n_{tc} \in\mathcal{W}_{tc,\mathrm{LO}}$, and $\check{\mathbf{w}}=(\mathbf{w}-\bar{\mathbf{w}})/\norm{\mathbf{w}-\bar{\mathbf{w}}}_1\in\mathcal{W}_{tc,\mathrm{CS}}$ with $\mathbf{w}\neq \bar{\mathbf{w}}$.  
Assume that  $\Delta(\bar{\mathbf{w}})>(1-\norm{\mathbf{w}-\bar{\mathbf{w}}}_1)\Delta(\mathbf{w})$. 
Then, we have $\Delta(\mathbf{w})>\frac{\Delta(\mathbf{w})-\Delta(\bar{\mathbf{w}})}{\norm{\mathbf{w}-\bar{\mathbf{w}}}_1}$. 
Since $\Delta(\cdot)$ is linear by its definition \cref{eq:def_port_delta}, $\frac{\Delta(\mathbf{w})-\Delta(\bar{\mathbf{w}})}{\norm{\mathbf{w}-\bar{\mathbf{w}}}_1}=\Delta(\frac{\mathbf{w}-\bar{\mathbf{w}}}{\norm{\mathbf{w}-\bar{\mathbf{w}}}_1})=\Delta(\check{\mathbf{w}})$. 
Thus, we have $\Delta(\mathbf{w})>\Delta(\check{\mathbf{w}})$. 

\textnormal{(ii)} 
For each $d\in \hat D_{tc}$, we have 
\begin{align}
 \Delta_{t-1,cd} & \ge \operatorname*{ess\,inf}\Delta_{t-1,cd}  \ge \Delta^{\min}_{t-1,c} \quad \text{a.s.},    \label{eq:delta_inf_min} \\
 \Delta_{t-1,cd} & \le \operatorname*{ess\,sup}\Delta_{t-1,cd}  \le \Delta^{\max}_{t-1,c} \quad \text{a.s.}    \label{eq:delta_sup_max}
\end{align}
Fix any $\mathbf{w}=[w_d]_{d \in \hat D_{tc}}\in\mathcal{W}_{tc,\mathrm{LO}}$. 
From \cref{eq:def_port_delta}, $\Delta(\mathbf{w})
=\sum_{d\in \hat D_{tc}} w_d \Delta_{t-1,cd}
\ge \sum_{d\in \hat D_{tc}} w_d\Delta^{\min}_{t-1,c}
=\Delta^{\min}_{t-1,c}$ a.s., 
where the inequality follows from \cref{eq:delta_inf_min}, 
and the last equality is because $\mathbf{w}\in\mathcal{W}_{tc,\mathrm{LO}}$. 
Therefore, we have 
\begin{equation}
    \operatorname*{ess\,inf}_{\mathbf{w}\in\mathcal{W}_{tc,\mathrm{LO}}}\Delta(\mathbf{w})\ge \Delta^{\min}_{t-1,c}. \label{eq:lo_delta_lb}
\end{equation}

Redefine $\mathbf{w}=[w_d]_{d \in \hat D_{tc}}\in\mathcal{W}_{tc,\mathrm{CS}}$. 
From \cref{eq:def_port_delta}, we obtain 
\begin{align*}
 \Delta(\mathbf{w})
&=\sum_{d\in \hat D_{tc}} [w_d]^+\Delta_{t-1,cd} -\sum_{d\in \hat D_{tc}} [-w_d]^+\Delta_{t-1,cd} \\
&\le \sum_{d\in \hat D_{tc}} [w_d]^+ \Delta^{\max}_{t-1,c}-\sum_{d\in \hat D_{tc}} [-w_d]^+\Delta^{\min}_{t-1,c} \quad\text{a.s.}\\
&=\tfrac12\Delta^{\max}_{t-1,c}-\tfrac12\Delta^{\min}_{t-1,c} \quad\text{a.s.},    
\end{align*}
where the first equality is because $w_d=[w_d]^+-[-w_d]^+$,
the inequality is due to \cref{eq:delta_inf_min,eq:delta_sup_max},
and the last equality is due to \cref{eq:symmetric_weight}.
Hence, we have 
\begin{equation}\label{eq:cs_delta_ub}
\operatorname*{ess\,sup}_{\mathbf{w}\in\mathcal{W}_{\mathrm{tc,CS}}}\Delta(\mathbf{w})
\le \tfrac{1}{2}({\Delta^{\max}_{t-1,c}-\Delta^{\min}_{t-1,c}}).
\end{equation}
If $3\Delta_{t-1,c}^{\min}>\Delta_{t-1,c}^{\max}$, then $\Delta_{t-1,c}^{\min}>\tfrac{1}{2}({\Delta_{t-1,c}^{\max}-\Delta_{t-1,c}^{\min}})$, yielding $\operatorname*{ess\,inf}_{\mathbf{w}\in\mathcal{W}_{tc,\mathrm{LO}}}\Delta(\mathbf{w})
>
\operatorname*{ess\,sup}_{\mathbf{w}\in\mathcal{W}_{tc,\mathrm{CS}}}\Delta(\mathbf{w})$
due to \cref{eq:lo_delta_lb,eq:cs_delta_ub}. 
\end{proof}

\setcounter{equation}{0}
\setcounter{figure}{0}
\setcounter{table}{0}
\section{CS Weight Projection via Group Demeaning} 
\label{appendix:group_neutralization}

% \red{$\hat{\mathbf{Y}}'_{t} \leftarrow \hat{\mathbf Y}^\circ$}

% \red{$\hat{\mathbf Y}^\circ \leftarrow  \hat{\mathbf Y}^\circ_t$}

% \red{$\bar{\hat{Y}}_{tc} \leftarrow  \hat{Y}^\bullet_{tc}$}

Given a prediction vector $\hat{\mathbf{Y}}_{t}=[\hat Y_{tcd}]_{(c,d)\in \hat U_t}\in\mathbb{R}^{\hat{n}_t}$, we obtain a projected CS weight vector $\mathbf{w}_t$ by computing:
\begin{equation}
    \mathbf{w}_t = \hat{\mathbf Y}^\circ_t/\|\hat{\mathbf Y}^\circ_t\|_1 \label{eq:demean_app1}
\end{equation}
where 
\begin{align}
    \hat{\mathbf Y}^\circ_t&=[\hat{Y}_{tcd}- \hat{Y}^\bullet_{tc}]_{(c,d)\in\hat{U}_t},  \label{eq:demean_app2}\\
 \hat{Y}^\bullet_{tc}&=\frac{1}{\hat{n}_{tc}} \sum_{d' \in \hat{D}_{tc}}  \hat{Y}_{tcd'}. \label{eq:demean_app3}    
\end{align} 
\Cref{eq:demean_app1} corresponds to a scaling that enforces $\| \mathbf{w}_t \|_1=1$. 
From the proof below, \cref{eq:demean_app2,eq:demean_app3} yield the projection of $\hat{\mathbf{Y}}_{t}$ onto $\{\mathbf{w}_t\in\mathbb{R}^{\hat{n}_t}:\mathbf 1^\top \mathbf{w}_{tc}=0,~\forall c\in \hat{C}_t\}$. 
Therefore, $\mathbf{w}_t$ is, up to scaling, the CS weight vector closest to $\hat{\mathbf{Y}}_{t}$.

We now show that \cref{eq:demean_app2,eq:demean_app3} are the Euclidean projection of $\hat{\mathbf{Y}}_{t}$ onto $\{\mathbf{w}_t\in\mathbb{R}^{\hat{n}_t}:\cref{eq:def_cs_multi}\}$. 
The Euclidean projection is given by 
\begin{align} 
\hat{\mathbf{Y}}'_{t}&=(I_{\hat{n}_t}-G_t^\top (G_t G_t^\top)^{-1} 
G_t)\hat{\mathbf{Y}}_{t}  \label{eq:projection}
\end{align} 
where $I_{\hat{n}_t}\in\mathbb{R}^{\hat{n}_t\times \hat{n}_t}$ is the identity matrix, 
$G_{t}=[\mathbf{g}_{tc}]_{c\in \hat{C}_t}\in\mathbb{R}^{|\hat{C}_t|\times \hat{n}_t}$, and $\mathbf{g}_{tc}=[\mathbb{I}_{\{c'=c\}}]_{(c',d')\in \hat{U}_t}\in \mathbb{R}^{\hat{n}_t}$,  
because \cref{eq:def_cs_multi} is equivalent to $G_t \mathbf{w}_t = 0 $. 
Note that
\begin{equation}
     \hat{Y}^\bullet_{tc} = \frac{1}{\hat{n}_{tc}} \mathbf{g}_{tc}^\top  \hat{\mathbf{Y}}_t. \label{eq:avg_app}
\end{equation}
The $(c',c'')$ element of $G_t G_t^\top$ is $\mathbf{g}_{tc'}^\top\mathbf{g}_{tc''}=\hat{n}_{tc'} \mathbb{I}_{\{c'=c''\}}$. 
%, which is $\hat{n}_{tc'}$ if $c'=c''$ and zero otherwise. 
From this, we have $G_t G_t^\top=\mathrm{diag}_{c\in C_t}(\hat{n}_{tc})$ and  
\begin{align}
 & G_t^\top(G_t G_t^\top)^{-1} G_t \hat{\mathbf{Y}}_t\\
 &=[\mathbf{g}_{tc_1}~ \mathbf{g}_{tc_2}~ \dots]\cdot\mathrm{diag}_{c\in C_t}(1/\hat{n}_{tc}) \cdot [\mathbf{g}_{tc_1}^\top; \mathbf{g}_{tc_2}^\top; \dots] \hat{\mathbf{Y}}_t\\
 &=\sum_{c\in \hat{C}_t} \frac{1}{\hat{n}_{tc}} \mathbf{g}_{tc}\mathbf{g}_{tc}^\top  \hat{\mathbf{Y}}_t \\
 &=\sum_{c\in \hat{C}_t}  \hat{Y}^\bullet_{tc} \mathbf{g}_{tc} \in \mathbb{R}^{\hat n_t} \label{eq:mean_commodity_for_futures}
\end{align}
by \cref{eq:avg_app}, where $c_1 < c_2$ are the two smallest elements in $\hat{C}_t$. 
The component of \cref{eq:mean_commodity_for_futures} corresponding to $(c,d)$ is the commodity-wise average $ \hat{Y}^\bullet_{tc}$ for $c$. 
From this and \cref{eq:projection}, the component of $\hat{\mathbf{Y}}'_{t}=(I_{\hat{n}_t}-G_t^\top (G_t G_t^\top)^{-1} G_t)\hat{\mathbf{Y}}_{t}$ corresponding to $(c,d)$ is $\hat Y_{tcd}- \hat{Y}^\bullet_{tc}$.
Hence, we have $\hat{\mathbf{Y}}'_{t}=[\hat{Y}_{tcd}- \hat{Y}^\bullet_{tc}]_{(c,d)\in\hat{U}_t}$, implying that \cref{eq:projection,eq:demean_app2} coincide. 
Therefore, \cref{eq:demean_app2,eq:demean_app3} are the Euclidean projection of $\hat{\mathbf{Y}}_{t}$ onto $\{\mathbf{w}_t\in\mathbb{R}^{\hat{n}_t}:\mathbf 1^\top \mathbf{w}_{tc}=0,~\forall c\in \hat{C}_t\}$.

\setcounter{equation}{0}
\setcounter{figure}{0}
\setcounter{table}{0}
\section{Data Processing}
\label{appendix:data_processing}

The maturity dates of futures contracts are initially estimated using their expiration dates, i.e., the \texttt{LastTrdDate} field in the \texttt{DSFutContrInfo} table of the LSEG Datastream database provided by WRDS. 
However, for some contracts, the estimated maturity dates based on \texttt{LastTrdDate} differ substantially from the expiration month recorded in the \texttt{ContrDate} field of the same table, which reports only the year and month rather than the exact date. 
To address this discrepancy, we further refine the estimated maturity dates by applying the following replacements in sequence:
\begin{itemize}
    \item \textbf{Replacement 1:} a valid \texttt{EndDate} field in the \texttt{DSFutContrChg} table,
    \item \textbf{Replacement 2:} a valid \texttt{MaxDate}, or
    \item \textbf{Replacement 3:} \texttt{AppContrDate}.
\end{itemize}
Following LSEG's recommendation, we use the \texttt{EndDate} field. 
\texttt{MaxDate} denotes the latest date on which the contract is traded in the price-volume table \texttt{DSFutContrVal}. 
Here, \texttt{AppContrDate} denotes the approximated maturity date. 
It is computed by adding the median difference between \texttt{LastTrdDate} and the 15th day of \texttt{ContrDate}, calculated from contracts with the same underlying commodity, to the 15th day of \texttt{ContrDate}. 
If the median is unavailable, we use the 15th day of \texttt{ContrDate}. 
For a value of \texttt{EndDate} and \texttt{MaxDate} to be valid, it should not be significantly different from \texttt{AppContrDate}. 
In addition, \texttt{EndDate} should not be less than \texttt{MaxDate} to be valid. 

\begin{table}
\caption{Summary of Maturity Date Replacements (Total Contracts: 26,461). The first row reports the number and ratio of issues before the replacements. Each subsequent row reports the number of replacements made and the remaining number of issues.}
\label{tab:maturity_replacement}
\begin{tabular}{lrrr}
\toprule
Replacement Type & Replacements & Remaining Issues & Ratio (\%) \\
\midrule
Original & 0 & 233 & 0.88 \\
Replacement 1 & 78 & 155 & 0.59 \\
Replacement 2 & 8 & 147 & 0.56 \\
Replacement 3 & 147 & 0 & 0.00 \\
\bottomrule
\end{tabular}
\end{table}

We use the \texttt{DSFutContrVal} table for the price-volume data, along with the \texttt{DSFutContrInfo} table. 
Following the indication of the \texttt{CurrUnitCode} field in the \texttt{DSFutContrInfo} table, we convert the unit of the price data. 
However, some contracts whose underlying commodity has other contracts requiring unit conversion, as indicated by \texttt{CurrUnitCode}, are not indicated for scaling but exhibit substantially different price scales. 
We therefore rescale these price data accordingly.

\setcounter{equation}{0}
\setcounter{figure}{0}
\setcounter{table}{0}
\section{Hyperparameter Selection Statistics}
\label{appendix:hyper_selection}
\begin{table}[t]
\caption{Hyperparameter Selection Counts under Yearly Cross-Validation }
\label{tab:selection_stats}
\begin{tabular}{llrrrrr}
\toprule
 &  & HGL & \ref{ablation:a1} & \ref{ablation:a2} & Individual & MLP \\
\midrule
\multirow[t]{2}{*}{\#params} & 10000 &  & 3 &  & 5 & 6 \\
 & 100000 & 10 & 7 & 10 & 5 & 4 \\
\cline{1-7}
\multirow[t]{3}{*}{$l_\mathrm{conv}$} & 1 &  & 1 & 7 & 1 &  \\
 & 2 & 3 & 1 & 3 & 4 & 6 \\
 & 3 & 7 & 8 &  & 5 & 4 \\
\cline{1-7}
\multirow[t]{3}{*}{$\rho^*$} & 0.1 & 10 &  & 9 & 4 &  \\
 & 0.2 &  &  & 1 & 2 &  \\
 & 0.3 &  &  &  & 4 &  \\
\cline{1-7}
\multirow[t]{3}{*}{CONV} & GAT & 3 &  & 1 &  &  \\
 & GCN & 2 & 10 & 2 &  &  \\
 & SAGE & 5 &  & 7 & 10 &  \\
\bottomrule
\end{tabular}
\end{table}
Each entry in \Cref{tab:selection_stats} reports the number of times a hyperparameter value is selected across the cross-validation periods.

\setcounter{equation}{0}
\setcounter{figure}{0}
\setcounter{table}{0}
\section{Trading Universe Characteristics}
\label{appendix:universe}
\Cref{fig:universe_tradability} summarizes the time-series characteristics of the trading universe used in our experiments.

\begin{figure*}
    \centering
    \includegraphics[width=1.0\linewidth]{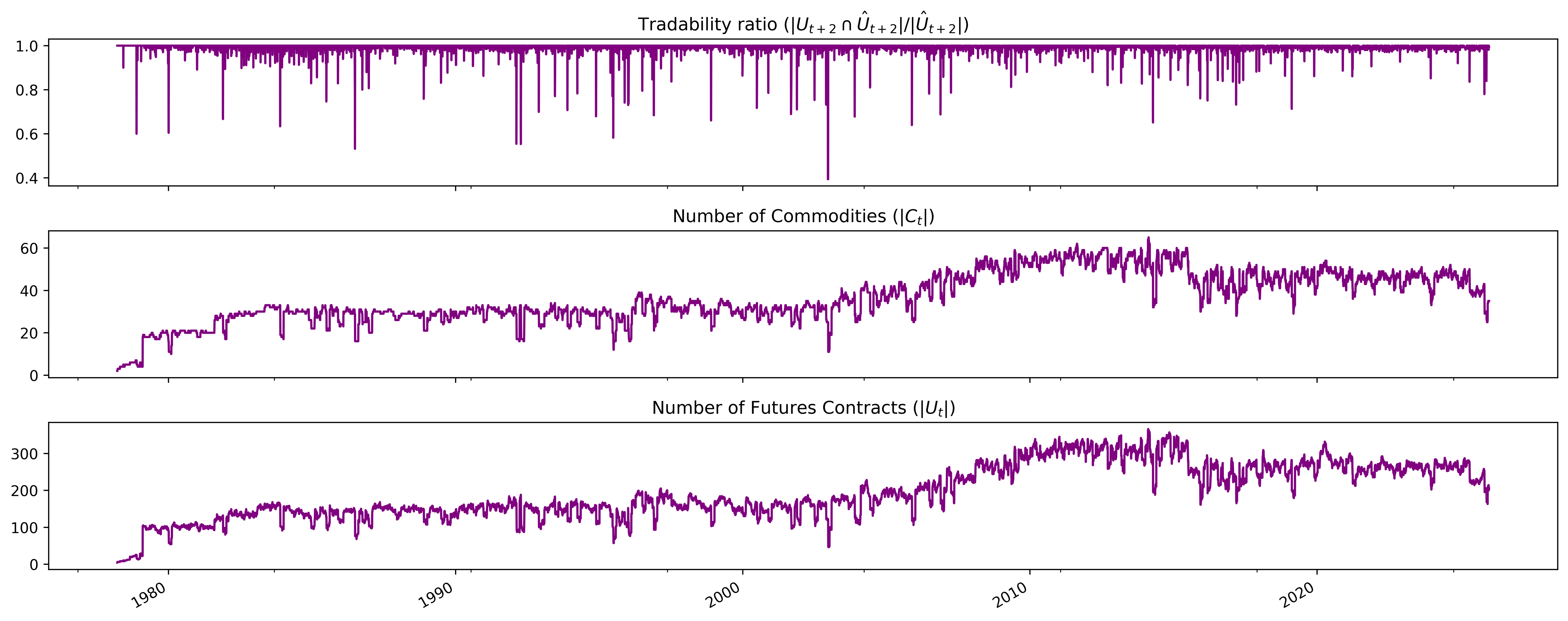}
    \caption{
    Time-Series Characteristics of the Trading Universe.
    }
    \label{fig:universe_tradability}
\end{figure*}

\setcounter{equation}{0}
\setcounter{figure}{0}
\setcounter{table}{0}
\section{Annual Trading Performance Metrics}
\label{appendix:annual}
\Cref{tab:annual_metrics1,tab:annual_metrics2} present annual trading performance metrics.
\begin{table}[t]
\caption{Annual Trading Performance Metrics (IR: information ratio; SR: Sortino ratio; Ret(\%): average daily return (\%); Vol(\%): daily return volatility (\%); Hit: hit ratio; MDD: maximum drawdown (absolute); Tvr: average daily turnover; Cor: correlation with the S\&P 500)}
\label{tab:annual_metrics1}
\resizebox{\linewidth}{!}{ \begin{tabular}{ll|rrrrrrr|rrrrrr}
\toprule
 &  & \multicolumn{7}{c}{CS} & \multicolumn{6}{c}{LO} \\
 &  & HGL & Ridge & MLP & LGBM & GNN & \ref{ablation:a1} & \ref{ablation:a2} & \ref{ablation:bm1} & \ref{ablation:bm2} & \ref{ablation:bm3} & \ref{ablation:bm4} & EW & S\&P 500 \\
Metric & Year &  &  &  &  &  &  &  &  &  &  &  &  &  \\
\midrule
\multirow[t]{10}{*}{IR} & 2016 & 0.0611 & 0.0110 & -0.0245 & 0.0017 & 0.0503 & 0.0547 & 0.0020 & 0.0592 & 0.0625 & 0.0649 & 0.0607 & 0.0587 & 0.0614 \\
 & 2017 & 0.1800 & -0.0047 & 0.0536 & 0.1394 & 0.0618 & 0.1413 & 0.0769 & 0.0186 & 0.0229 & 0.0260 & 0.0257 & 0.0181 & 0.1712 \\
 & 2018 & 0.1211 & -0.0132 & 0.0567 & 0.0746 & -0.0292 & 0.1412 & 0.0514 & -0.0394 & -0.0343 & -0.0341 & -0.0326 & -0.0398 & -0.0324 \\
 & 2019 & -0.0174 & 0.0036 & -0.0587 & -0.0299 & -0.0663 & -0.0111 & -0.0258 & 0.0356 & 0.0372 & 0.0377 & 0.0343 & 0.0354 & 0.1474 \\
 & 2020 & 0.1214 & 0.0448 & 0.0494 & 0.1049 & 0.0797 & 0.1004 & 0.0979 & 0.0343 & 0.0387 & 0.0399 & 0.0419 & 0.0339 & 0.0367 \\
 & 2021 & 0.0185 & 0.0100 & -0.0964 & -0.0888 & -0.1063 & 0.0181 & -0.0603 & 0.1468 & 0.1461 & 0.1454 & 0.1489 & 0.1469 & 0.1261 \\
 & 2022 & 0.1123 & 0.0243 & 0.0756 & 0.0525 & 0.0244 & 0.0785 & 0.0566 & 0.0618 & 0.0639 & 0.0700 & 0.0681 & 0.0612 & -0.0495 \\
 & 2023 & 0.0478 & 0.0656 & 0.0599 & 0.0482 & 0.0359 & 0.0456 & 0.0763 & -0.0605 & -0.0612 & -0.0575 & -0.0586 & -0.0607 & 0.1009 \\
 & 2024 & 0.0906 & 0.0605 & 0.0146 & 0.0225 & -0.0023 & 0.0982 & 0.0168 & -0.0218 & -0.0152 & -0.0129 & -0.0169 & -0.0225 & 0.1202 \\
 & 2025 & 0.1435 & 0.1225 & 0.1179 & 0.0997 & 0.0781 & 0.1337 & 0.1228 & 0.0423 & 0.0513 & 0.0546 & 0.0490 & 0.0416 & 0.0541 \\
\cline{1-15}
\multirow[t]{10}{*}{SR} & 2016 & 0.0909 & 0.0150 & -0.0360 & 0.0024 & 0.0697 & 0.0723 & 0.0028 & 0.1058 & 0.1109 & 0.1161 & 0.1084 & 0.1050 & 0.0804 \\
 & 2017 & 0.2980 & -0.0070 & 0.0969 & 0.2430 & 0.0979 & 0.2433 & 0.1216 & 0.0295 & 0.0370 & 0.0420 & 0.0412 & 0.0287 & 0.2296 \\
 & 2018 & 0.1719 & -0.0149 & 0.0798 & 0.1149 & -0.0352 & 0.2136 & 0.0797 & -0.0523 & -0.0458 & -0.0455 & -0.0436 & -0.0529 & -0.0395 \\
 & 2019 & -0.0236 & 0.0066 & -0.0778 & -0.0414 & -0.0824 & -0.0124 & -0.0334 & 0.0529 & 0.0554 & 0.0563 & 0.0508 & 0.0526 & 0.1871 \\
 & 2020 & 0.1591 & 0.0561 & 0.0715 & 0.1692 & 0.1123 & 0.1412 & 0.1396 & 0.0402 & 0.0462 & 0.0476 & 0.0495 & 0.0397 & 0.0417 \\
 & 2021 & 0.0280 & 0.0154 & -0.1471 & -0.1299 & -0.1421 & 0.0283 & -0.0917 & 0.1908 & 0.1920 & 0.1905 & 0.1951 & 0.1908 & 0.1787 \\
 & 2022 & 0.2017 & 0.0303 & 0.1477 & 0.0998 & 0.0434 & 0.1316 & 0.0933 & 0.0822 & 0.0873 & 0.0961 & 0.0921 & 0.0813 & -0.0793 \\
 & 2023 & 0.0719 & 0.0963 & 0.1092 & 0.0898 & 0.0605 & 0.0723 & 0.1279 & -0.0959 & -0.0963 & -0.0906 & -0.0919 & -0.0962 & 0.1665 \\
 & 2024 & 0.1340 & 0.0830 & 0.0241 & 0.0330 & -0.0036 & 0.1383 & 0.0244 & -0.0353 & -0.0248 & -0.0210 & -0.0274 & -0.0363 & 0.1585 \\
 & 2025 & 0.2469 & 0.2250 & 0.1766 & 0.1751 & 0.1406 & 0.2114 & 0.2058 & 0.0589 & 0.0712 & 0.0764 & 0.0691 & 0.0577 & 0.0680 \\
\cline{1-15}
\multirow[t]{10}{*}{Ret (\%)} & 2016 & 0.0043 & 0.0012 & -0.0019 & 0.0002 & 0.0055 & 0.0037 & 0.0002 & 0.0449 & 0.0459 & 0.0477 & 0.0461 & 0.0446 & 0.0505 \\
 & 2017 & 0.0101 & -0.0004 & 0.0046 & 0.0096 & 0.0052 & 0.0076 & 0.0056 & 0.0093 & 0.0112 & 0.0130 & 0.0129 & 0.0090 & 0.0723 \\
 & 2018 & 0.0120 & -0.0022 & 0.0071 & 0.0091 & -0.0056 & 0.0134 & 0.0063 & -0.0217 & -0.0187 & -0.0188 & -0.0180 & -0.0220 & -0.0351 \\
 & 2019 & -0.0013 & 0.0004 & -0.0050 & -0.0028 & -0.0077 & -0.0008 & -0.0025 & 0.0198 & 0.0201 & 0.0206 & 0.0191 & 0.0197 & 0.1136 \\
 & 2020 & 0.0213 & 0.0091 & 0.0127 & 0.0229 & 0.0212 & 0.0143 & 0.0212 & 0.0343 & 0.0372 & 0.0388 & 0.0421 & 0.0339 & 0.0797 \\
 & 2021 & 0.0023 & 0.0016 & -0.0120 & -0.0113 & -0.0159 & 0.0019 & -0.0077 & 0.1228 & 0.1206 & 0.1204 & 0.1238 & 0.1230 & 0.1033 \\
 & 2022 & 0.0231 & 0.0059 & 0.0188 & 0.0129 & 0.0065 & 0.0151 & 0.0125 & 0.0786 & 0.0774 & 0.0860 & 0.0856 & 0.0780 & -0.0754 \\
 & 2023 & 0.0048 & 0.0089 & 0.0075 & 0.0070 & 0.0043 & 0.0043 & 0.0093 & -0.0463 & -0.0448 & -0.0428 & -0.0448 & -0.0465 & 0.0833 \\
 & 2024 & 0.0111 & 0.0093 & 0.0017 & 0.0033 & -0.0003 & 0.0096 & 0.0020 & -0.0143 & -0.0097 & -0.0085 & -0.0112 & -0.0148 & 0.0959 \\
 & 2025 & 0.0160 & 0.0182 & 0.0153 & 0.0138 & 0.0122 & 0.0137 & 0.0141 & 0.0328 & 0.0389 & 0.0426 & 0.0383 & 0.0322 & 0.0640 \\
\cline{1-15}
\multirow[t]{10}{*}{Vol (\%)} & 2016 & 0.0708 & 0.1053 & 0.0796 & 0.0917 & 0.1096 & 0.0679 & 0.0776 & 0.7580 & 0.7350 & 0.7362 & 0.7598 & 0.7600 & 0.8211 \\
 & 2017 & 0.0560 & 0.0852 & 0.0860 & 0.0690 & 0.0847 & 0.0536 & 0.0723 & 0.4986 & 0.4913 & 0.4977 & 0.5004 & 0.4988 & 0.4221 \\
 & 2018 & 0.0990 & 0.1682 & 0.1251 & 0.1216 & 0.1901 & 0.0951 & 0.1232 & 0.5518 & 0.5472 & 0.5503 & 0.5512 & 0.5522 & 1.0831 \\
 & 2019 & 0.0755 & 0.1016 & 0.0849 & 0.0931 & 0.1156 & 0.0748 & 0.0956 & 0.5556 & 0.5404 & 0.5455 & 0.5579 & 0.5566 & 0.7711 \\
 & 2020 & 0.1752 & 0.2041 & 0.2567 & 0.2180 & 0.2654 & 0.1420 & 0.2162 & 0.9990 & 0.9597 & 0.9708 & 1.0046 & 1.0016 & 2.1705 \\
 & 2021 & 0.1251 & 0.1583 & 0.1242 & 0.1276 & 0.1496 & 0.1039 & 0.1273 & 0.8360 & 0.8257 & 0.8282 & 0.8320 & 0.8370 & 0.8190 \\
 & 2022 & 0.2057 & 0.2423 & 0.2490 & 0.2451 & 0.2651 & 0.1919 & 0.2211 & 1.2709 & 1.2109 & 1.2292 & 1.2573 & 1.2742 & 1.5232 \\
 & 2023 & 0.1014 & 0.1360 & 0.1245 & 0.1450 & 0.1205 & 0.0947 & 0.1222 & 0.7647 & 0.7316 & 0.7449 & 0.7647 & 0.7667 & 0.8261 \\
 & 2024 & 0.1222 & 0.1530 & 0.1182 & 0.1474 & 0.1373 & 0.0982 & 0.1200 & 0.6548 & 0.6422 & 0.6593 & 0.6639 & 0.6549 & 0.7974 \\
 & 2025 & 0.1112 & 0.1486 & 0.1300 & 0.1386 & 0.1568 & 0.1022 & 0.1148 & 0.7745 & 0.7583 & 0.7797 & 0.7803 & 0.7745 & 1.1835 \\
\cline{1-15}
\multirow[t]{10}{*}{MDD} & 2016 & 0.0126 & 0.0154 & 0.0151 & 0.0144 & 0.0154 & 0.0090 & 0.0145 & 0.0905 & 0.0852 & 0.0825 & 0.0896 & 0.0912 & 0.0820 \\
 & 2017 & 0.0046 & 0.0144 & 0.0103 & 0.0059 & 0.0075 & 0.0066 & 0.0063 & 0.0746 & 0.0673 & 0.0686 & 0.0727 & 0.0751 & 0.0281 \\
 & 2018 & 0.0108 & 0.0300 & 0.0177 & 0.0190 & 0.0346 & 0.0066 & 0.0142 & 0.1186 & 0.1087 & 0.1103 & 0.1115 & 0.1194 & 0.2145 \\
 & 2019 & 0.0084 & 0.0104 & 0.0186 & 0.0144 & 0.0258 & 0.0101 & 0.0126 & 0.0881 & 0.0798 & 0.0783 & 0.0870 & 0.0889 & 0.0699 \\
 & 2020 & 0.0187 & 0.0230 & 0.0319 & 0.0286 & 0.0237 & 0.0144 & 0.0186 & 0.2884 & 0.2724 & 0.2737 & 0.2842 & 0.2896 & 0.3825 \\
 & 2021 & 0.0142 & 0.0216 & 0.0289 & 0.0277 & 0.0414 & 0.0152 & 0.0199 & 0.0668 & 0.0682 & 0.0674 & 0.0668 & 0.0667 & 0.0527 \\
 & 2022 & 0.0200 & 0.0507 & 0.0205 & 0.0267 & 0.0441 & 0.0209 & 0.0267 & 0.1862 & 0.1683 & 0.1677 & 0.1749 & 0.1876 & 0.2562 \\
 & 2023 & 0.0083 & 0.0169 & 0.0100 & 0.0152 & 0.0150 & 0.0089 & 0.0095 & 0.1539 & 0.1440 & 0.1447 & 0.1557 & 0.1547 & 0.1065 \\
 & 2024 & 0.0097 & 0.0152 & 0.0171 & 0.0145 & 0.0204 & 0.0072 & 0.0180 & 0.1362 & 0.1319 & 0.1291 & 0.1325 & 0.1368 & 0.0872 \\
 & 2025 & 0.0077 & 0.0160 & 0.0091 & 0.0094 & 0.0168 & 0.0088 & 0.0075 & 0.0939 & 0.0886 & 0.0893 & 0.0924 & 0.0942 & 0.2041 \\
\cline{1-15}
\multirow[t]{10}{*}{Hit} & 2016 & 0.5397 & 0.5119 & 0.5119 & 0.5159 & 0.5198 & 0.5357 & 0.5000 & 0.5198 & 0.5159 & 0.5238 & 0.5238 & 0.5198 & 0.5238 \\
 & 2017 & 0.5720 & 0.4680 & 0.5080 & 0.5680 & 0.5360 & 0.5600 & 0.5320 & 0.5360 & 0.5480 & 0.5480 & 0.5440 & 0.5360 & 0.5720 \\
 & 2018 & 0.5595 & 0.4960 & 0.4921 & 0.5040 & 0.4802 & 0.5714 & 0.5079 & 0.5040 & 0.5198 & 0.5159 & 0.5040 & 0.5040 & 0.5198 \\
 & 2019 & 0.4841 & 0.4722 & 0.5119 & 0.5040 & 0.5000 & 0.5198 & 0.5437 & 0.5357 & 0.5357 & 0.5278 & 0.5397 & 0.5357 & 0.5952 \\
 & 2020 & 0.5573 & 0.5296 & 0.5336 & 0.5573 & 0.5296 & 0.5375 & 0.5613 & 0.5573 & 0.5613 & 0.5534 & 0.5573 & 0.5573 & 0.5731 \\
 & 2021 & 0.5198 & 0.5278 & 0.4603 & 0.4603 & 0.4405 & 0.5119 & 0.4762 & 0.5833 & 0.5754 & 0.5754 & 0.5754 & 0.5833 & 0.5675 \\
 & 2022 & 0.5179 & 0.5020 & 0.4861 & 0.4861 & 0.4861 & 0.5179 & 0.5020 & 0.5697 & 0.5498 & 0.5578 & 0.5657 & 0.5737 & 0.4303 \\
 & 2023 & 0.5280 & 0.5320 & 0.5160 & 0.5240 & 0.5200 & 0.5280 & 0.5400 & 0.4760 & 0.4720 & 0.4640 & 0.4680 & 0.4760 & 0.5440 \\
 & 2024 & 0.5397 & 0.5278 & 0.5119 & 0.5159 & 0.4960 & 0.5317 & 0.5119 & 0.5159 & 0.5119 & 0.5079 & 0.5119 & 0.5119 & 0.5714 \\
 & 2025 & 0.5783 & 0.5020 & 0.5823 & 0.5261 & 0.5301 & 0.5301 & 0.5382 & 0.5141 & 0.5100 & 0.5100 & 0.5181 & 0.5141 & 0.5743 \\
\bottomrule
\end{tabular}}
\end{table}

\begin{table}[t]
\caption{Annual Trading Performance Metrics (IR: information ratio; SR: Sortino ratio; Ret(\%): average daily return (\%); Vol(\%): daily return volatility (\%); Hit: hit ratio; MDD: maximum drawdown (absolute); Tvr: average daily turnover; Cor: correlation with the S\&P 500)}
\label{tab:annual_metrics2}
\resizebox{\linewidth}{!}{ \begin{tabular}{ll|rrrrrrr|rrrrrr}
\toprule
 &  & \multicolumn{7}{c}{CS} & \multicolumn{6}{c}{LO} \\
 &  & HGL & Ridge & MLP & LGBM & GNN & \ref{ablation:a1} & \ref{ablation:a2} & \ref{ablation:bm1} & \ref{ablation:bm2} & \ref{ablation:bm3} & \ref{ablation:bm4} & EW & S\&P 500 \\
Metric & Year &  &  &  &  &  &  &  &  &  &  &  &  &  \\
\midrule
\multirow[t]{10}{*}{Tvr} & 2016 & 0.5276 & 0.6454 & 0.4692 & 0.6413 & 0.5683 & 0.5884 & 0.5114 & 0.0582 & 0.2663 & 0.2512 & 0.2262 & 0.0547 &  \\
 & 2017 & 0.5583 & 0.6547 & 0.6300 & 0.7193 & 0.6138 & 0.5790 & 0.5902 & 0.0554 & 0.2882 & 0.2566 & 0.2318 & 0.0520 &  \\
 & 2018 & 0.6625 & 0.6705 & 0.4946 & 0.6540 & 0.7327 & 0.5235 & 0.5890 & 0.0500 & 0.3197 & 0.2444 & 0.2257 & 0.0431 &  \\
 & 2019 & 0.6274 & 0.4144 & 0.5628 & 0.6654 & 0.6017 & 0.5541 & 0.5415 & 0.0524 & 0.3078 & 0.2735 & 0.2481 & 0.0469 &  \\
 & 2020 & 0.5313 & 0.6429 & 0.5680 & 0.6852 & 0.5764 & 0.5615 & 0.5032 & 0.0472 & 0.2592 & 0.2532 & 0.2231 & 0.0430 &  \\
 & 2021 & 0.6155 & 0.6048 & 0.4857 & 0.6504 & 0.5687 & 0.4720 & 0.4963 & 0.0510 & 0.2984 & 0.2353 & 0.2180 & 0.0467 &  \\
 & 2022 & 0.5350 & 0.7176 & 0.5349 & 0.5913 & 0.5300 & 0.6061 & 0.5649 & 0.0483 & 0.2646 & 0.2492 & 0.2262 & 0.0466 &  \\
 & 2023 & 0.5039 & 0.8716 & 0.4988 & 0.6557 & 0.5084 & 0.4757 & 0.5332 & 0.0462 & 0.2595 & 0.2120 & 0.1948 & 0.0433 &  \\
 & 2024 & 0.5400 & 0.4508 & 0.4357 & 0.6471 & 0.6230 & 0.4480 & 0.5234 & 0.0471 & 0.2696 & 0.2576 & 0.2159 & 0.0426 &  \\
 & 2025 & 0.5149 & 0.6623 & 0.5022 & 0.6805 & 0.6556 & 0.5749 & 0.5621 & 0.0496 & 0.2667 & 0.2881 & 0.2324 & 0.0462 &  \\
\cline{1-15}
\multirow[t]{10}{*}{Cor} & 2016 & 0.0300 & 0.0181 & 0.1579 & 0.1066 & 0.1552 & 0.0461 & 0.1137 & 0.3874 & 0.3834 & 0.3816 & 0.3892 & 0.3878 & 1.0000 \\
 & 2017 & 0.0426 & -0.0565 & 0.0582 & 0.0358 & 0.0161 & 0.0129 & 0.0352 & 0.0387 & 0.0288 & 0.0319 & 0.0405 & 0.0391 & 1.0000 \\
 & 2018 & 0.0235 & 0.1229 & 0.0838 & 0.1115 & 0.0506 & 0.0588 & 0.1405 & 0.3021 & 0.3035 & 0.3048 & 0.3027 & 0.3017 & 1.0000 \\
 & 2019 & -0.0898 & -0.0621 & -0.0365 & -0.0151 & -0.0259 & 0.0345 & -0.0353 & 0.3595 & 0.3561 & 0.3538 & 0.3554 & 0.3599 & 1.0000 \\
 & 2020 & 0.0267 & 0.1062 & 0.0783 & 0.0943 & 0.0663 & -0.0394 & 0.0140 & 0.5252 & 0.5189 & 0.5224 & 0.5257 & 0.5252 & 1.0000 \\
 & 2021 & -0.0747 & -0.0416 & -0.0602 & -0.1115 & -0.0496 & -0.1307 & -0.0737 & 0.2850 & 0.2754 & 0.2806 & 0.2833 & 0.2854 & 1.0000 \\
 & 2022 & -0.0512 & 0.0293 & 0.0167 & 0.0275 & 0.0216 & -0.0444 & -0.0025 & 0.1496 & 0.1509 & 0.1492 & 0.1483 & 0.1497 & 1.0000 \\
 & 2023 & 0.0499 & 0.0573 & 0.0062 & 0.0155 & 0.0741 & 0.0283 & 0.0493 & 0.2752 & 0.2799 & 0.2789 & 0.2779 & 0.2748 & 1.0000 \\
 & 2024 & -0.0012 & 0.0017 & 0.0170 & -0.0009 & 0.0515 & 0.0102 & 0.0144 & 0.1219 & 0.1272 & 0.1261 & 0.1196 & 0.1218 & 1.0000 \\
 & 2025 & -0.1068 & 0.0305 & -0.1755 & -0.0624 & 0.0287 & -0.1505 & -0.0576 & 0.3426 & 0.3353 & 0.3294 & 0.3340 & 0.3434 & 1.0000 \\
\bottomrule
\end{tabular}}
\end{table}
\clearpage
\label{LastPage}
\addtocounter{page}{-1} %I added

\end{document}